\documentclass[amsart,11pt,oneside,english]{article}
\usepackage{amssymb,amsmath,amsfonts, amsmath,mathtools, 
amsthm,euscript,mathrsfs
,verbatim
,enumerate,multirow
,bbding,slashed
,multicol,color,array, 
esint,babel, tikz,tikz-cd,tikz-3dplot,tkz-graph,pgfplots,ytableau,graphicx,float,rotating,hyperref,geometry,mathdots,savesym,cite, wasysym,amscd,graphicx,pifont,float,setspace,multirow,wrapfig,picture,subfigure, amsthm,hepth, enumitem, enumerate,xcolor
}
\usepackage{dynkin-diagrams}
\usepackage{array}% http://ctan.org/pkg/array
\usepackage[utf8]{inputenc}
\usepackage{thm-restate}
 \usetikzlibrary{math}
\usepackage{thmtools}
\usepackage{prettyref}
\usepackage[T1]{fontenc}
\usepackage{datetime}
\numberwithin{equation}{section}
\numberwithin{figure}{section}
\usetikzlibrary{arrows,positioning,decorations.pathmorphing,   decorations.markings, matrix, patterns,shapes}
\hypersetup{colorlinks=true}
\hypersetup{linkcolor=black}
\hypersetup{citecolor=black}
\hypersetup{urlcolor=black}
\theoremstyle{plain}
\newtheorem*{thm*}{Theorem}
\theoremstyle{plain}% default
\newtheorem{thm}{Theorem}[section]

\newtheorem{lem}[thm]{Lemma}

\theoremstyle{definition}
\newtheorem{defn}[thm]{Definition}
\newtheorem*{defn*}{Definition}

\newtheorem{rem}[thm]{Remark}

\usepackage[normalem]{ulem}

\tikzset{
  big arrow/.style={
    decoration={markings,mark=at position 1 with {\arrow[scale=1.5,#1]{>}}},
    postaction={decorate},
    shorten >=0.4pt},
  big arrow/.default=black}

\usepackage{makecell}%To keep spacing of text in tables

\numberwithin{equation}{section}

	\definecolor{myblue}{rgb}{.72,.83,.97}
	\definecolor{myred}{rgb}{0.9, 0.44, 0.37}
\begin{document}
\begin{titlepage}
\begin{center}
\vspace{4cm}
{\Huge\bfseries   D$_4$-flops  of the E$_7$-model \\  }
\vspace{2cm}
{%
\LARGE  Mboyo Esole$^{\spadesuit}$ and Sabrina Pasterski$^\dagger$\\}
\vspace{1cm}

{\large $^{\spadesuit}$ Department of Mathematics, Northeastern University}\par
{  360 Huntington Avenue, Boston, MA 02115, USA}\par
 \scalebox{.95}{Email: \quad {\tt  j.esole@northeastern.edu }}\par
\vspace{.3cm}
{\large $^\dagger$ Center for the Fundamental Laws of Nature,  Harvard University}\par
{ 17 Oxford Street, Cambridge, MA 02138, USA}\par
 \scalebox{.95}{Email:\quad {\tt  spasterski@fas.harvard.edu }}\par
\vspace{3cm}
{ \bf{Abstract:}}\\
\end{center}
{\date{\today\  \currenttime}}

We study  the geography of  crepant resolutions of E$_7$-models.  
An E$_7$-model is a Weierstrass model corresponding to the output of  Step 9 of Tate's algorithm characterizing the  Kodaira fiber of type III$^*$  over the generic point of a smooth  prime divisor. The dual graph of the Kodaira fiber of type III$^*$ is the affine Dynkin diagram of type $\widetilde{\text{E}}_7$.  
A Weierstrass model of type  E$_7$  is conjectured to have eight distinct crepant resolutions whose flop diagram is a Dynkin diagram of type E$_8$. 
We construct explicitly four of these eight crepant resolutions forming a  sub-diagram of type D$_4$. We explain how the flops between these four crepant resolutions can be understood using the flops between the crepant resolutions of  two well-chosen suspended pinch points. 
\vfill 
%\noindent{\today\par}


\end{titlepage}

\tableofcontents

\newpage 

\section{Introduction}
This paper aims to explore the  geography of crepant resolutions of  E$_7$-models given by 
singular Weierstrass models corresponding to Step 9 of Tate's algorithm \cite{Tate,Bershadsky:1996nh,Katz:2011qp}.   
In M-theory and F-theory compactifications, such geometries produce E$_7$ gauge theories (respectively in five- and  six-dimensional spacetime with eight supersymmetric generators), with 
matter transforming in the adjoint and the fundamental representation of E$_7$ of dimension 56. We will denote that representation as $\bf{56}$ in the rest of the paper. 
    Flops between distinct crepant resolutions are then interpreted as phase transitions between distinct Coulomb chambers of the ${\cal N}=1$ five-dimensional gauge theory.   
        The intersection polynomial is not invariant under flops and has to be computed in each of the crepant resolutions.

The last few years have seen a deep improvement of our understanding of the geography of crepant resolutions of Weierstrass models.  
 Crepant resolutions of a singular Weierstrass model are relative minimal models (in the sense of Mori) over the Weierstrass model and are connected to each other by a finite sequence of flops. 
            One significant incomplete problem at the boundary between birational geometry and string geometry  is the explicit construction of all the crepant resolutions of a given Weierstrass model coming from Tate's algorithm \cite{Anderson:2017zfm}. 
The geography of these crepant resolutions is  the study of the network of flops connecting  them. 
 It is natural to ask what is the geography of the minimal models corresponding to a given Weierstrass model:  
 How many such minimal models are there? 
 What is the graph of their flops?\footnote{ We attach a graph to the collection of minimal models of a Weierstrass model such that  the nodes of the graph are in bijection with the minimal models and two nodes are connected by an edge when the two corresponding minimal models are connected by a flop. 
 Such a graph is called the {\em graph of flops of the minimal models}. 
 If the elliptic fibration is a Calabi-Yau threefold, the graph of flops corresponds to the incidence graph of the distinct chambers of the Coulomb branch of the five-dimensional gauge theory obtained by a compactification of M-theory on the elliptic fibration.  
 }

    The systematic investigation of the geography of minimal models of Weierstrass models started with the study of the crepant resolutions of the  SU($5$)-model and has been completed for $G$-models with $G$ a simple complex Lie group of low-rank such as SU($n$) (for $n=2,3,4,5$) \cite{ESY1,ESY2},  G$_2$,  Spin($7$), Spin($8$) \cite{G2}, and F$_4$ \cite{F4}.    
Semi-simple cases and cases with non-trivial Mordell--Weil groups have also been investigated recently \cite{Esole:2017hlw,SU2G2, SO4}. 
There are still some important omissions. 
With the exception of infinite series, E$_5$=Spin($10$), E$_6$ and E$_7$  are the two crucial cases left for which the details of the crepant resolutions defining each chamber are still a mystery.  The E$_6$-model  corresponds to a Kodaira fiber IV$^{*\text{s}}$  covered by Step 8 of Tate's algorithm while the E$_7$-model corresponds to the Kodaira fiber III$^*$ and step 9 of Tate's algorithm.

We would like to start  a detailed exploration of the   minimal models of the E$_7$-model.  In this context, minimal models are  crepant resolutions over the Weierstrass model.  
We would like to  explicitly construct each minimal model over the Weierstrass model as a  projective variety defined by a crepant resolution of the Weierstrass model of an E$_7$-model and study the geography of these different crepant resolutions.  
     An E$_7$ Weierstrass  model is conjectured to have eight distinct crepant resolutions whose flop diagram is a Dynkin diagram of type E$_8$. However, to this day, these crepant resolutions have not been identified.   
    We will construct four of the eight conjectured minimal models of an E$_7$-model and show that their flops define a Dynkin diagram of type D$_4$.
    In our construction, the birational geometry of the suspended pinch point will be used to model the flops of the minimal models we discuss.

\subsection{Defining the E$_7$-model} 

An E$_7$-model is an elliptic fibration $Y\longrightarrow B$ over a smooth base $B$ with a choice of a  smooth prime divisor $S$ in the base  $B$  such that the  fiber over the generic point of $S$ is of Kodaira type III$^*$ 
 and all singular fibers over generic points of the discriminant locus away from $S$ are irreducible (of Kodaira type  II or I$_1$). 
  The name ``E$_7$-model'' stems from the fact that the  dual graph of a Kodaira fiber \cite{Kodaira} of type III$^*$ is the affine Dynkin diagram of type $\widetilde{\text{E}}_7$. 
E$_7$-models are  typically given by singular Weierstrass models using Step 9 of Tate's algorithm, which can be  traced to Proposition 4 of N\'eron's thesis \cite{Neron}.

Let $B$ be a smooth variety  of dimension two or higher over the complex numbers. 
Given a line bundle $\mathscr{L}$, we define the projective bundle of lines $\pi: X_0=\mathbb{P}_B[\mathscr{O}_B\oplus \mathscr{L}^{\otimes 2}\oplus \mathscr{L}^{\otimes 3}]\to B$. 
The projective bundle $X_0$  is the ambient space for a Weierstrass model \cite{Esole:2017csj}. 
Let $\mathscr{O}_{X_0}(1)$ be the dual of the tautological line bundle of $X_0$, a Weierstrass model is by definition the zero scheme of a section of the line bundle $\mathscr{O}_{X_0}(3)\otimes \pi^* \mathscr{L}^{\otimes 6}$. 
Throughout this paper, we work over the complex numbers $\mathbb{C}$.   Given sections $f_i$ of lines bundles $\mathscr{L}_i$, we denote by $V(f_1, \ldots, f_r)$ their vanishing scheme defined as the zero scheme $f_1=f_2=\cdots=f_r=0$.
If we denote the relative projective coordinates of $X_0$ as $[z:x:y]$, a Weierstrass model is the following vanishing locus 
$$
V(y^2z-x^3- f xz^2 -g z^3),
$$
where $f$ is a section of $\mathscr{L}^{\otimes 4}$ and $g$ is  a section of $\mathscr{L}^{\otimes 6}$.
The discriminant  and the $j$-invariant are given by the following expressions 
 $$
\Delta = 4 f^3 +27 g^2, \quad j=1728 \frac{4f^3}{\Delta}.
$$
The divisor $\Delta$ is a  section  of the line bundle  $\mathscr{L}^{\otimes 12}$.  
The locus of points of $B$ over which the fiber is singular is $V(\Delta)$. The fibers of a smooth Weierstrass model are all irreducible of type I$_0$ (smooth elliptic curve), I$_1$ (nodal elliptic curve), and II (cuspidal elliptic curve). 
Reducible fibers appear only after resolving the singularities of a singular Weierstrass model up to codimension-two. 

The general Weierstrass equation is \cite{Formulaire}
\begin{equation}
y^2z +a_1 x y z + a_3 y z^2- x^3 -a_2 x^2 z -a_{4}x z^2 -a_6 z^3=0. 
\end{equation}
 Let $S=V(s)$ be a smooth and irreducible Cartier divisor in a smooth variety $B$.  
We denote the  generic point of $S$ and its residue field by $\eta$ and $\kappa$, respectively. 
 We define a valuation $v_S$ such that $v_S(f)=n$ when a rational function $f$ has a zero of multiplicity $n$ if $n\geq 0$ or a pole of multiplicity $-n$ if $n< 0$. 

An E$_7$-model is  characterized by  Step  9 of Tate's algorithm and corresponds to  type (c7) in N\'eron's classification \cite{Neron}. Following  Proposition 4 of \cite{Neron}, an E$_7$-model is characterized by the  following restrictions on the valuations of the coefficients of the Weierstrass equation: 
\begin{equation}
v_S(a_1)\geq 1, \quad v_S(a_2)\geq 2, \quad v_S(a_3)\geq 3, \quad v_S(a_4)=3, \quad v_S(a_6)\geq 5.
\end{equation}
After completing the square in $y$ and the cube in $x$, an elliptic fibration with generic fiber of Kodaira type III$^*$  over $S=V(s)$ can always be written as the following Weierstrass model 
\begin{equation}
y^2z = x^3 + a_{4, 3} s^3 x z^2 + a_{6, 5+\beta} s^{5+\beta} z^3, \quad \beta\in \mathbb{Z}_{\geq 0},
\end{equation}
where $a_{i,j}$ is a section of $\mathscr{L}^{\otimes i} \otimes \mathscr{S}^{-\otimes j}$ where $\mathscr{S}=\mathscr{O}_B(S)$. 
Such an elliptic fibration is called an E$_7$-model. 
In this paper, we focus on the case $\beta=0$ and, to  ease the notation, we will write the Weierstrass model as follows 
\begin{equation}\label{eq:E7}
y^2z = x^3 + a s^3 x z^2 + b s^5 z^3,
\end{equation}
where $S=V(s)$ is a smooth Cartier divisor defined by the zero locus of a section of the line bundle $\mathscr{S}$,  $a$ is a section of $\mathscr{L}^{\otimes 4} \otimes \mathscr{S}^{-\otimes 3}$, 
and $b$ is a section of $\mathscr{L}^{\otimes 6} \otimes \mathscr{S}^{-\otimes 5}$. 
We assume that $a$ and $b$ have zero valuation along $S$ and $V(a)$ and $V(b)$ are smooth divisors in $B$ intersecting transversally. 

The discriminant locus of this Weierstrass model is the vanishing scheme of $\Delta$ with 
\begin{equation}
\Delta=s^9 (4a^3 + 27 b^2 s).
\end{equation}
The reduced discriminant locus consists of two prime divisors, namely 
\begin{equation}
S=V(s)\quad  \text{ and }\quad 
\Delta'= V(4a^3 + 27b^2 s).
\end{equation}
The divisor $\Delta'$ has cuspidal singularities at $V(a,b)$ that worsen to triple point singularities over $V(a,b,s)$. The divisors $S$ and $\Delta'$ do not intersect transversally as their intersection scheme consists of triple points  $(s,a^3)$. 
Since the generic point of $\Delta'$ is regular, we can still apply Tate's algorithm along $\Delta'$. 
The fiber over the generic point of $S$ has Kodaira type  III$^*$ and the generic fiber over $\Delta'$ is of Kodaira type I$_1$. 
The collision of singularities III$^*$+I$_1$ is not in the list of Miranda as the two fibers have distinct $j$-invariant. 
We do expect a degeneration of the fiber III$^*$ over the  the codimension-two locus $V(s,a)$ of the base $B$ and further at the codimension-three loci $V(s,a,b)$. 

\subsection{E$_7$ facts and conjectures}
In this section, we recall some  facts and conjectures about the  E$_7$-model that are relevant to appreciate the questions addressed in this paper. 
\begin{enumerate}
 \item {\bf Topological invariants.} 
 The Hodge numbers and the Euler characteristic of a crepant resolution are invariant under flops and can therefore be computed in any crepant resolution. 
 The Euler characteristic over a base of arbitrary dimension of an E$_7$-model has been computed in \cite{Euler}. The Hodge numbers and the Euler characteristic of an E$_7$-model in the special case of an elliptically fibered Calabi-Yau threefold that is obtained by a crepant resolution of a Weierstrass model are independent of the choice of crepant resolution \cite{Batyrev.Betti} and are given in \cite{Euler}. 
 The key point is that the Euler characteristic is  
$$\chi=-60K^2 -98K\cdot S - 42S^2=-60 K^2 +196 (1-g)+56 S^2.$$
 The characteristic numbers of an E$_7$-model are computed in \cite{EK.Invariant}. See also \cite{EKY,EFY}.

\item {\bf Singular fibers of fibral divisors and the minuscule  representation $\bf{56}$.}  The fiber of type III$^*$ degenerates over the collision of $S$ and $\Delta'$. The irreducible rational curves forming the singular fiber over $S\cap \Delta'$  have geometric weights in the representation $\mathbf{56}$ of E$_7$, which is a  minuscule fundamental representation.
 Thus, the matter representation associated with an E$_7$-model is the direct sum of the  adjoint ($\bf{133}$)  and the fundamental representation ($\mathbf{56}$) of E$_7$. 
 \item {\bf Non-Higgsable cluster.}
 The representation $\mathbf{56}$ does not occur when  $S$ and $\Delta'$ do not intersect. This is famously illustrated by the non-Higgsable cluster corresponding to the local Calabi-Yau threefold defined over the quasi-projective surface given by the total space of the line bundle $\mathscr{O}_{\mathbb{P}^1}(-8)$. See \cite{Morrison:2012np,DelZotto:2017pti,Bhardwaj:2018yhy}.

\item {\bf Loss of flatness.}  When the base of the fibration is at least of dimension three, a  crepant resolution of the E$_7$-model does not give a flat fibration over the base as there are codimension-three points  over which the fiber contains a full quadric surface as discussed in the partial toric resolutions of \cite{Candelas:2000nc}.  We will see that this happens exactly over the intersection of $S$ with the  singular locus of $\Delta'$, namely over the codimension-three locus $V(s,a,b)$ in the base $B$. 
\item {\bf The geography of crepant resolutions.}  The authors of \cite{Box} conjectured that there are eight crepant resolutions of the Weierstrass model of an E$_7$-model and the graph of their flops is a Dynkin diagram of type E$_8$. 
This is based on studying the hyperplane arrangement defined by the weights of the representation $\mathbf{56}$ inside the dual fundamental Weyl chamber of E$_7$. 
The crepant resolutions were not constructed explicitly. See also \cite{Diaconescu:1998cn}.
\item {\bf The fiber structure.}  
 The authors of \cite{Box} also conjectured  that the crepant resolution corresponding to the $\alpha_i$ of E$_8$ is such that the  generic fiber over $S$  degenerates over $V(a)$ 
 to a  non-Kodaira fiber whose dual graph is  the same as the affine Dynkin diagram of type $\widetilde{\text{E}}_8$ with the node $\alpha_i$ contracted to a point. 
 This has not been verified geometrically in more than one chamber. 
 \end{enumerate}

  \subsection{Summary of results}
  The key results of this paper are the following: 
 \begin{enumerate}
\item  We  give four distinct crepant resolutions of the Weierstrass model of the E$_7$-model. 
We show that the graph of the flops between these four crepant resolutions is a 
D$_4$-Dynkin subdiagram of the expected E$_8$ flop-diagram as illustrated in Figure \ref{fig:IE756}. 
More specifically the resolutions correspond to the nodes  $\alpha_4$, $\alpha_5$, $\alpha_ 6$, and $\alpha_8$ of the Dynkin diagram of type E$_8$ of  Figure \ref{fig:IE756}.  
For that reason, we call these resolutions Y$_4$, Y$_5$, Y$_6$, and Y$_8$, respectively. 
We show that the  D$_4$-Dynkin diagram of flops can be  understood by the flops of two suspended pinch points as illustrated on Figure \ref{Fig:A3A4D4}. 

\item We show that in the resolution Y$_i$ ($i=4,5,6,8$), the generic fiber over $V(s,a)$ corresponds to the Dynkin diagram of type $\widetilde{\text{E}}_8$ with the node $\alpha_i$ contracted to a point. 
\item For each resolution, we identify the new curves that appear over $V(s,a)$. They are extreme rays of the  resolution and their geometric weights are  always in the minuscule representation $\mathbf{56}$. 
In particular, we find the five extreme rays corresponding to the weights  $\varpi_{23}$, $\varpi_{26}$, $\varpi_{29}$, $\varpi_{30}$, $\varpi_{32}$ as illustrated on  Figure \ref{fig:IE756}. 
 \item We also show that in each of the four crepant resolutions studied in this paper, the divisor $D_6$ corresponding to the root $\alpha_6$ has a fiber that jumps in dimension over the codimension-three locus $V(s,a,b)$. All other fibers are always one  dimensional over any point of the base. 

\item We compute the triple intersection numbers of the fibral divisors and in this way give a geometric derivation of the Chern-Simons levels and of the superpotential of an E$_7$-model in the four chambers studied in this paper. 
\item We compute the number of hypermultiplets in an M-theory compactification on a Calabi-Yau threefold that is an E$_7$-model by comparing the triple intersection numbers with the one-loop corrected superpotential of a five-dimensional supergravity theory with Lie group E$_7$ and matter transforming in the adjoint and the fundamental representation $\bf{56}$ \cite{IMS}.
As expected, the number of adjoints is the genus of the curve supporting the fiber III$^*$ and 
the number of fundamentals is given by the number of intersection points between the two irreducible components 
of the reduced discriminant locus:
 $$n_{\bf{133}}=g, \quad n_{\bf{56}}=\frac{1}{2} S\cdot [a]=\frac{1}{2}\Big(8 (1-g)+S^2\Big).$$
But this intersection is not transverse as $S$ and $\Delta'$ intersect along the scheme $(s, a^3)$ which consists 
of triple points supported on $V(s,a)$. 
 The factor of $1/2$ indicates that  each intersection point contributes one half-hypermultiplet,   which is possible because the representation $\bf{56}$ is pseudo-real. 
 When $S$ is a rational curve, $n_{\bf{133}}=0$. If $S$ is a rational curve  of self-intersection $S^2=-8$, we get n$_{\bf{56}}=0$ while when it is  a rational curve of self-intersection $S^2=-7$ we get n$_{\bf{56}}=1/2$, which corresponds to a unique half-hypermultiplet.
 \end{enumerate}

For convenience, we collect our results in the following pages.

\clearpage

\begin{figure}[b]
\begin{center}

 \begin{tikzpicture}[scale=.9]
 \node[draw,circle,thick,scale=1] (0) at (0,0){$Ch_1$};
				\node[draw,circle,thick,scale=1] (1) at (1.2*1.8,0){$Ch_2$};
				\node[draw,circle,thick,scale=1] (2) at (2.4*1.8,0){$Ch_3$};
				\node[draw,circle,thick,scale=1, fill=lightgray] (3) at (3.6*1.8,0){$Ch_4$};
				\node[draw,circle,thick,scale=1,fill=lightgray] (4) at (4.8*1.8,0){$Ch_5$};
				\node[draw,circle,thick,scale=1,fill=lightgray] (5) at (6*1.8,0){$Ch_6$};
				\node[draw,circle,thick,scale=1] (6) at (7.2*1.8,0){$Ch_7$};
				\node[draw,circle,thick,scale=1,fill=lightgray] (7) at (4.8*1.8,1.2*1.8){$Ch_8$};
				\draw[thick] (0)-- node[above] {$\varpi_{19}$} (1)--node[above] {$\varpi_{20}$} (2)--node[above] {$\varpi_{23}$} (3)--node[above] {$\varpi_{26}$} (4)--node[above] {$\varpi_{29}$} (5)--node[above] {$\varpi_{32}$} (6);
				\draw[thick]  (4)--   node[right] {$\varpi_{30}$}  (7);
					\end{tikzpicture}
 \end{center}
 \caption{Adjacency  graph of the chambers of the hyperplane arrangement I$(E_7, \mathbf{56})$. Each chamber is simplicial. 
   A weight $\varpi$ between two nodes indicates that the corresponding  chambers  are separated by the hyperplane $\varpi^\bot$.  For example, one goes from Ch$_1$ to Ch$_2$ by crossing the hyperplane $\varpi_{19}^\bot$. 
 The colored chambers forming a subgraph of type D$_4$ are those corresponding to the nef-cone of the crepant resolutions constructed in this paper. 
    \label{fig:IE756}}
\end{figure}
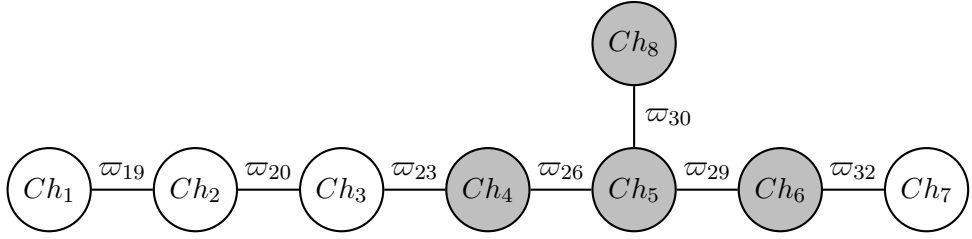

\begin{table}[hbt]
\begin{center}
\scalebox{.9}{
\begin{tabular}{|c|c|}
\hline
\begin{tikzpicture}[scale=.8]
				% E8
				\node at (1.4,2){$Ch_1$};
				\node[draw,circle,thick,scale=1,label=below:{\scalebox{1.2}{$\varpi_{19}$}}] (1) at (1.2,0){$-$};
				\node[draw,circle,thick,scale=1,label=below:{\scalebox{1.2}{$\varpi_{20}$}}] (2) at (2.4,0){$-$};
				\node[draw,circle,thick,scale=1,label=below:{\scalebox{1.2}{$\varpi_{23}$}}] (3) at (3.6,0){$-$};
				\node[draw,circle,thick,scale=1,label=below:{\scalebox{1.2}{$\varpi_{26}$}}] (4) at (4.8,0){$-$};
				\node[draw,circle,thick,scale=1,label=below:{\scalebox{1.2}{$\varpi_{29}$}}] (5) at (6,0){$-$};
				\node[draw,circle,thick,scale=1,label=below:{\scalebox{1.2}{$\varpi_{32}$}}] (6) at (7.2,0){$-$};
				\node[draw,circle,thick,scale=1, label=above:{\scalebox{1.2}{$\varpi_{30}$}}] (8) at (4.8,1.6){$-$};
				\draw[thick] (1)--(2)--(3)--(4)--(5)--(6);
				\draw[thick]  (4)--(8);
					\end{tikzpicture}& 			
					\begin{tikzpicture}[scale=.9]
				% E8
				\node at (1.4,2){$Ch_2$};
				\node[draw,circle,thick,scale=1,label=below:{\scalebox{1.2}{$\varpi_{19}$}}] (1) at (1.2,0){$+$};
				\node[draw,circle,thick,scale=1,label=below:{\scalebox{1.2}{$\varpi_{20}$}}] (2) at (2.4,0){$-$};
				\node[draw,circle,thick,scale=1,label=below:{\scalebox{1.2}{$\varpi_{23}$}}] (3) at (3.6,0){$-$};
				\node[draw,circle,thick,scale=1,label=below:{\scalebox{1.2}{$\varpi_{26}$}}] (4) at (4.8,0){$-$};
				\node[draw,circle,thick,scale=1,label=below:{\scalebox{1.2}{$\varpi_{29}$}}] (5) at (6,0){$-$};
				\node[draw,circle,thick,scale=1,label=below:{\scalebox{1.2}{$\varpi_{32}$}}] (6) at (7.2,0){$-$};
				\node[draw,circle,thick,scale=1, label=above:{\scalebox{1.2}{$\varpi_{30}$}}] (8) at (4.8,1.6){$-$};
				\draw[thick] (1)--(2)--(3)--(4)--(5)--(6);
				\draw[thick]  (4)--(8);
					\end{tikzpicture}	
					\\	\hline			
					\begin{tikzpicture}[scale=.9]
				% E8
				\node at (1.4,2){$Ch_3$};
				\node[draw,circle,thick,scale=1,label=below:{\scalebox{1.2}{$\varpi_{19}$}}] (1) at (1.2,0){$+$};
				\node[draw,circle,thick,scale=1,label=below:{\scalebox{1.2}{$\varpi_{20}$}}] (2) at (2.4,0){$+$};
				\node[draw,circle,thick,scale=1,label=below:{\scalebox{1.2}{$\varpi_{23}$}}] (3) at (3.6,0){$-$};
				\node[draw,circle,thick,scale=1,label=below:{\scalebox{1.2}{$\varpi_{26}$}}] (4) at (4.8,0){$-$};
				\node[draw,circle,thick,scale=1,label=below:{\scalebox{1.2}{$\varpi_{29}$}}] (5) at (6,0){$-$};
				\node[draw,circle,thick,scale=1,label=below:{\scalebox{1.2}{$\varpi_{32}$}}] (6) at (7.2,0){$-$};
				\node[draw,circle,thick,scale=1, label=above:{\scalebox{1.2}{$\varpi_{30}$}}] (8) at (4.8,1.6){$-$};
				\draw[thick] (1)--(2)--(3)--(4)--(5)--(6);
				\draw[thick]  (4)--(8);
					\end{tikzpicture}& 	
					\begin{tikzpicture}[scale=.9]
				% E8
				\node at (1.4,2){$Ch_4$};
				\node[draw,circle,thick,scale=1,label=below:{\scalebox{1.2}{$\varpi_{19}$}}] (1) at (1.2,0){$+$};
				\node[draw,circle,thick,scale=1,label=below:{\scalebox{1.2}{$\varpi_{20}$}}] (2) at (2.4,0){$+$};
				\node[draw,circle,thick,scale=1,label=below:{\scalebox{1.2}{$\varpi_{23}$}}] (3) at (3.6,0){$+$};
				\node[draw,circle,thick,scale=1,label=below:{\scalebox{1.2}{$\varpi_{26}$}}] (4) at (4.8,0){$-$};
				\node[draw,circle,thick,scale=1,label=below:{\scalebox{1.2}{$\varpi_{29}$}}] (5) at (6,0){$-$};
				\node[draw,circle,thick,scale=1,label=below:{\scalebox{1.2}{$\varpi_{32}$}}] (6) at (7.2,0){$-$};
				\node[draw,circle,thick,scale=1, label=above:{\scalebox{1.2}{$\varpi_{30}$}}] (8) at (4.8,1.6){$-$};
				\draw[thick] (1)--(2)--(3)--(4)--(5)--(6);
				\draw[thick]  (4)--(8);
					\end{tikzpicture}
					\\ \hline
						\begin{tikzpicture}[scale=.9]
				% E8
				\node at (1.4,2){$Ch_5$};
				\node[draw,circle,thick,scale=1,label=below:{\scalebox{1.2}{$\varpi_{19}$}}] (1) at (1.2,0){$+$};
				\node[draw,circle,thick,scale=1,label=below:{\scalebox{1.2}{$\varpi_{20}$}}] (2) at (2.4,0){$+$};
				\node[draw,circle,thick,scale=1,label=below:{\scalebox{1.2}{$\varpi_{23}$}}] (3) at (3.6,0){$+$};
				\node[draw,circle,thick,scale=1,label=below:{\scalebox{1.2}{$\varpi_{26}$}}] (4) at (4.8,0){$+$};
				\node[draw,circle,thick,scale=1,label=below:{\scalebox{1.2}{$\varpi_{29}$}}] (5) at (6,0){$-$};
				\node[draw,circle,thick,scale=1,label=below:{\scalebox{1.2}{$\varpi_{32}$}}] (6) at (7.2,0){$-$};
				\node[draw,circle,thick,scale=1, label=above:{\scalebox{1.2}{$\varpi_{30}$}}] (8) at (4.8,1.6){$-$};
				\draw[thick] (1)--(2)--(3)--(4)--(5)--(6);
				\draw[thick]  (4)--(8);
					\end{tikzpicture}&					\begin{tikzpicture}[scale=.9]
				% E8
				\node at (1.4,2){$Ch_6$};
				\node[draw,circle,thick,scale=1,label=below:{\scalebox{1.2}{$\varpi_{19}$}}] (1) at (1.2,0){$+$};
				\node[draw,circle,thick,scale=1,label=below:{\scalebox{1.2}{$\varpi_{20}$}}] (2) at (2.4,0){$+$};
				\node[draw,circle,thick,scale=1,label=below:{\scalebox{1.2}{$\varpi_{23}$}}] (3) at (3.6,0){$+$};
				\node[draw,circle,thick,scale=1,label=below:{\scalebox{1.2}{$\varpi_{26}$}}] (4) at (4.8,0){$+$};
				\node[draw,circle,thick,scale=1,label=below:{\scalebox{1.2}{$\varpi_{29}$}}] (5) at (6,0){$+$};
				\node[draw,circle,thick,scale=1,label=below:{\scalebox{1.2}{$\varpi_{32}$}}] (6) at (7.2,0){$-$};
				\node[draw,circle,thick,scale=1, label=above:{\scalebox{1.2}{$\varpi_{30}$}}] (8) at (4.8,1.6){$-$};
				\draw[thick] (1)--(2)--(3)--(4)--(5)--(6);
				\draw[thick]  (4)--(8);
					\end{tikzpicture}
					\\\hline

					\begin{tikzpicture}[scale=.9]
				% E8
				\node at (1.4,2){$Ch_7$};
				\node[draw,circle,thick,scale=1,label=below:{\scalebox{1.2}{$\varpi_{19}$}}] (1) at (1.2,0){$+$};
				\node[draw,circle,thick,scale=1,label=below:{\scalebox{1.2}{$\varpi_{20}$}}] (2) at (2.4,0){$+$};
				\node[draw,circle,thick,scale=1,label=below:{\scalebox{1.2}{$\varpi_{23}$}}] (3) at (3.6,0){$+$};
				\node[draw,circle,thick,scale=1,label=below:{\scalebox{1.2}{$\varpi_{26}$}}] (4) at (4.8,0){$+$};
				\node[draw,circle,thick,scale=1,label=below:{\scalebox{1.2}{$\varpi_{29}$}}] (5) at (6,0){$+$};
				\node[draw,circle,thick,scale=1,label=below:{\scalebox{1.2}{$\varpi_{32}$}}] (6) at (7.2,0){$+$};
				\node[draw,circle,thick,scale=1, label=above:{\scalebox{1.2}{$\varpi_{30}$}}] (8) at (4.8,1.6){$-$};
				\draw[thick] (1)--(2)--(3)--(4)--(5)--(6);
				\draw[thick]  (4)--(8);
					\end{tikzpicture}& 
					\begin{tikzpicture}[scale=.9]
				% E8
				\node at (1.4,2){$\displaystyle{Ch_8}$};
				\node[draw,circle,thick,scale=1,label=below:{\scalebox{1.2}{$\varpi_{19}$}}] (1) at (1.2,0){$+$};
				\node[draw,circle,thick,scale=1,label=below:{\scalebox{1.2}{$\varpi_{20}$}}] (2) at (2.4,0){$+$};
				\node[draw,circle,thick,scale=1,label=below:{\scalebox{1.2}{$\varpi_{23}$}}] (3) at (3.6,0){$+$};
				\node[draw,circle,thick,scale=1,label=below:{\scalebox{1.2}{$\varpi_{26}$}}] (4) at (4.8,0){$+$};
				\node[draw,circle,thick,scale=1,label=below:{\scalebox{1.2}{$\varpi_{29}$}}] (5) at (6,0){$-$};
				\node[draw,circle,thick,scale=1,label=below:{\scalebox{1.2}{$\varpi_{32}$}}] (6) at (7.2,0){$-$};
				\node[draw,circle,thick,scale=1, label=above:{\scalebox{1.2}{$\varpi_{30}$}}] (8) at (4.8,1.6){$+$};
				\draw[thick] (1)--(2)--(3)--(4)--(5)--(6);
				\draw[thick]  (4)--(8);
					\end{tikzpicture}
\\
\hline					\end{tabular}}
					\end{center}
					\caption{The eight chambers of I(E$_7$, $\bf{56})$. Each chamber is uniquely defined by the signs taken by the seven linear functions $\langle \varpi_i, \phi\rangle$ for $i=\{19, 20,23,26,29,32,30\}$. Together, they define a sign vector for the hyperplane arrangement. These weights satisfy a partial order whose Hasse diagram is a Dynkin diagram of type E$_7$ as illustrated in Figure \ref{Fig:SV}. 
					\label{Table:Ch}}

\end{table}
\clearpage

\newpage

\begin{table}[b!]
\begin{center}
\newcolumntype{C}{ >{\centering\arraybackslash} m{8cm} }
\scalebox{1.4}{\begin{tabular}{|c|C|}
\hline
\text{Chamber 4}
 & 
\scalebox{.8}{
\begin{tikzpicture}
				% E8
				\node[draw,circle,thick,scale=1,fill=black,label=below:{\scalebox{1.2}{ $\alpha_0$}}] (0) at (0,0){$1$};
				\node[draw,circle,thick,scale=1,label=below:{\scalebox{1.2}{$\alpha_1$}}] (1) at (1.2,0){$2$};
				\node[draw,circle,thick,scale=1,label=below:{\scalebox{1.2}{$\alpha_2$}}] (2) at (2.4,0){$3$};
				\node[draw,circle,thick,scale=1,label=below:{\scalebox{1.2}{$\alpha_3$}}] (3) at (3.6,0){$4$};
				\node[draw,circle,thick,scale=1,label=below:{\scalebox{1.2}{$\alpha_5$}}] (5) at (6,0){$6$};
				\node[draw,circle,thick,scale=1,label=below:{\scalebox{1.2}{$\alpha_6$}}] (6) at (7.2,0){$4$};
				\node[draw,circle,thick,scale=1, label=below:{\scalebox{1.2}{$\alpha_7$}}] (7) at (8.4,0){$2$};
				\node[draw,circle,thick,scale=1, label=above:{\scalebox{1.2}{$\alpha_8$}}] (8) at (6,1.6){$3$};
				\draw[thick] (0)--(1)--(2)--(3)--(5)--(6)--(7);
				\draw[thick]  (5)--(8);
					\end{tikzpicture}	
			}

\\
\hline 
Chamber 5 & 
\scalebox{.8}{
\begin{tikzpicture}
				% E8
				\node[draw,circle,thick,scale=1,fill=black,label=below:{\scalebox{1.2}{ $\alpha_0$}}] (0) at (0,0){$1$};
				\node[draw,circle,thick,scale=1,label=below:{\scalebox{1.2}{$\alpha_1$}}] (1) at (1.2,0){$2$};
				\node[draw,circle,thick,scale=1,label=below:{\scalebox{1.2}{$\alpha_2$}}] (2) at (2.4,0){$3$};
				\node[draw,circle,thick,scale=1,label=below:{\scalebox{1.2}{$\alpha_3$}}] (3) at (3.6,0){$4$};
				\node[draw,circle,thick,scale=1,label=below:{\scalebox{1.2}{$\alpha_4$}}] (4) at (4.8,0){$5$};
			 				\node[draw,circle,thick,scale=1,label=below:{\scalebox{1.2}{$\alpha_6$}}] (6) at (7.2,0){$4$};
				\node[draw,circle,thick,scale=1, label=below:{\scalebox{1.2}{$\alpha_7$}}] (7) at (8.4,0){$2$};
				\node[draw,circle,thick,scale=1, label=above:{\scalebox{1.2}{$\alpha_8$}}] (8) at (6,1.6){$3$};
				\draw[thick] (0)--(1)--(2)--(3)--(4)--(6)--(7);
				\draw[thick]  (6,0)--(8);
					\end{tikzpicture}	
					}
					\\
				\hline
Chamber 6 & 
		\scalebox{.8}{
		\begin{tikzpicture}
				% E8
				\node[draw,circle,thick,scale=1,fill=black,label=below:{\scalebox{1.2}{ $\alpha_0$}}] (0) at (0,0){$1$};
				\node[draw,circle,thick,scale=1,label=below:{\scalebox{1.2}{$\alpha_1$}}] (1) at (1.2,0){$2$};
				\node[draw,circle,thick,scale=1,label=below:{\scalebox{1.2}{$\alpha_2$}}] (2) at (2.4,0){$3$};
				\node[draw,circle,thick,scale=1,label=below:{\scalebox{1.2}{$\alpha_3$}}] (3) at (3.6,0){$4$};
				\node[draw,circle,thick,scale=1,label=below:{\scalebox{1.2}{$\alpha_4$}}] (4) at (4.8,0){$5$};
				\node[draw,circle,thick,scale=1,label=below:{\scalebox{1.2}{$\alpha_5$}}] (5) at (6,0){$6$};
								\node[draw,circle,thick,scale=1, label=below:{\scalebox{1.2}{$\alpha_7$}}] (7) at (8.4,0){$2$};
				\node[draw,circle,thick,scale=1, label=above:{\scalebox{1.2}{$\alpha_8$}}] (8) at (6,1.6){$3$};
				\draw[thick] (0)--(1)--(2)--(3)--(4)--(5)--(7);
				\draw[thick]  (5)--(8);
					\end{tikzpicture}}
					\\
					\hline & \\
								Chamber 8 & 	
					\scalebox{.8}{
					\begin{tikzpicture}
				% E8
				\node[draw,circle,thick,scale=1,fill=black,label=below:{\scalebox{1.2}{ $\alpha_0$}}] (0) at (0,0){$1$};
				\node[draw,circle,thick,scale=1,label=below:{\scalebox{1.2}{$\alpha_1$}}] (1) at (1.2,0){$2$};
				\node[draw,circle,thick,scale=1,label=below:{\scalebox{1.2}{$\alpha_2$}}] (2) at (2.4,0){$3$};
				\node[draw,circle,thick,scale=1,label=below:{\scalebox{1.2}{$\alpha_3$}}] (3) at (3.6,0){$4$};
				\node[draw,circle,thick,scale=1,label=below:{\scalebox{1.2}{$\alpha_4$}}] (4) at (4.8,0){$5$};
				\node[draw,circle,thick,scale=1,label=below:{\scalebox{1.2}{$\alpha_5$}}] (5) at (6,0){$6$};
				\node[draw,circle,thick,scale=1,label=below:{\scalebox{1.2}{$\alpha_6$}}] (6) at (7.2,0){$4$};
				\node[draw,circle,thick,scale=1, label=below:{\scalebox{1.2}{$\alpha_7$}}] (7) at (8.4,0){$2$};
							\draw[thick] (0)--(1)--(2)--(3)--(4)--(5)--(6)--(7);
				%\draw[thick]  (5)--(8);
				
					\end{tikzpicture}}
					\\
					\hline

	\end{tabular}	
	}
	\caption{Degeneration of the E$_7$-fiber over $V(s,a)$.  These fibers should be compared with the affine Dynkin diagram of type $\widetilde{\text{E}}_8$ in Figure \ref{Fig:E8}.
	The degenerated fiber is a Kodaira fiber of type II$^*$ with a node contracted to a point. 
	In each chamber, the node that is contracted is different. 
	These fibers are computed directly from explicit crepant resolutions of singularities. 
	 \label{Figure:NK}
	 }	
		\end{center}						
					
	\end{table}

\clearpage

\clearpage

\begin{table}[htb]
\begin{center}
{\setlength{\extrarowheight}{30pt} \begin{tabular}{|l|}
\hline
$
 \begin{aligned}
6{\cal F}_1(\phi) &=2(1-n_A) \left(4 \phi _1^3+3 \phi _2 \phi _1^2-6 \phi _2^2 \phi _1+4 \phi _2^3+4 \phi _3^3+4 \phi _4^3+4 \phi _5^3+4 \phi _6^3+4 \phi _7^3\right. \\
&\quad \left.-3 \phi _2 \phi _3^2+3 \phi _4 \phi _5^2+6 \phi _5 \phi _6^2-3 \phi _3^2 \phi _4-6 \phi _4^2 \phi _5-9 \phi _5^2 \phi _6-3 \phi _3^2 \phi _7\right)\\
&\quad -2n_F \phi _6 \left(6 \phi _1^2-6 \phi _2 \phi _1+6 \phi _2^2+6 \phi _3^2+6 \phi _4^2\right. \\
& \quad \left. +6 \phi _5^2+5 \phi _6^2+6 \phi _7^2-6 \phi _2 \phi _3-6 \phi _3 \phi _4-6 \phi _4 \phi _5-6 \phi _5 \phi _6-6 \phi _3 \phi _7\right)\\[1ex]
\end{aligned}
$
\\
\hline 
$
\begin{aligned}
6{\cal F}_2(\phi) & = 
2(1-n_A) \left(4 \phi _1^3+3 \phi _2 \phi _1^2-6 \phi _2^2 \phi _1+4 \phi _2^3+4 \phi _3^3+4 \phi _4^3+4 \phi _5^3+4 \phi _6^3\right. \\
& \quad \left. 
+4 \phi _7^3-3 \phi _2 \phi _3^2
+3 \phi _4 \phi _5^2
+6 \phi _5 \phi _6^2-3 \phi _3^2 \phi _4-6 \phi _4^2 \phi _5-9 \phi _5^2 \phi _6-3 \phi _3^2 \phi _7\right) \\
& \quad -2 n_F\left(\phi _1^3+3 \phi _6 \phi _1^2+3 \left(\phi _6-2 \phi _2\right) \phi _6 \phi _1    +2 \phi _6 \left(3 \phi _2^2-3 \phi _3 \phi _2+\right.\right. \\
& \quad \left. \left. 3 \phi _3^2+2 \phi _6^2+3 \phi _7^2+3 \left(\phi _4^2-\phi _5 \phi _4+\phi _5^2\right)-3 \phi _5 \phi _6-3 \phi _3 (\phi _4+\phi _7)      \right)\right)\\[1ex]
\end{aligned}$\\
\hline
$ \begin{aligned}
6{\cal F}_3(\phi)  &=2 \left(1-n_A\right)
(4 \phi _1^3+3 \phi _2 \phi _1^2-6 \phi _2^2 \phi _1+4 \phi _2^3+4 \phi _3^3+4 \phi _4^3+4 \phi _5^3+4 \phi _6^3\\
&\quad + 4 \phi _7^3-3 \phi _2 \phi _3^2+3 \phi _4 \phi _5^2+6 \phi _5 \phi _6^2-3 \phi _3^2 \phi _4-6 \phi _4^2 \phi _5-9 \phi _5^2 \phi _6-3 \phi _3^2 \phi _7)\\
& \quad 
-2n_F\big(\phi _2^3-3 \phi _1 \phi _2^2+3 \phi _6 \phi _2^2+3 \phi _1^2 \phi _2+3 (\phi _6-2 \phi _3) \phi _6 \phi _2\\
& \quad +3 \phi _6(2 \phi _3^2-2 (\phi _4+\phi _7) \phi _3+\phi _6^2+2 \phi _7^2+2 (\phi _4^2-\phi _5 \phi _4+\phi _5^2)-2 \phi _5 \phi _6)\big) \\[1ex]
\end{aligned}
$
\\
\hline
$
 \begin{aligned}
6{\cal F}_4(\phi) & =2 (1-n_A) (4 \phi _1^3+3 \phi _2 \phi _1^2-6 \phi _2^2 \phi _1+4 \phi _2^3+4 \phi _3^3+4 \phi _4^3+4 \phi _5^3\\
& \quad+4 \phi _6^3+4 \phi _7^3-3 \phi _2 \phi _3^2+3 \phi _4 \phi _5^2+6 \phi _5 \phi _6^2-3 \phi _3^2 \phi _4-6 \phi _4^2 \phi _5-9 \phi _5^2 \phi _6-3 \phi _3^2 \phi _7  ) \\
&\quad -2 n_F (\phi _3^3-3 \phi _2 \phi _3^2+3 \phi _6 \phi _3^2+3 \phi _2^2 \phi _3+3 \phi _6^2 \phi _3-6 \phi _4 \phi _6 \phi _3-6 \phi _6 \phi _7 \phi _3+2 \phi _6^3\\
&\quad -3 \phi _1 \phi _2^2-6 \phi _5 \phi _6^2+6 \phi _6 \phi _7^2+3 \phi _1^2 \phi _2+6 \phi _4^2 \phi _6+6 \phi _5^2 \phi _6-6 \phi _4 \phi _5 \phi _6  )\\[1ex]
\end{aligned}
$
\\
\hline
$
 \begin{aligned}
6{\cal F}_5(\phi) & = 
2   (1-n_A)    (4 \phi _1^3+3 \phi _2 \phi _1^2-6 \phi _2^2 \phi _1+4 \phi _2^3+4 \phi _3^3+4 \phi _4^3+4 \phi _5^3+4 \phi _6^3\\
& \quad+4 \phi _7^3-3 \phi _2 \phi _3^2+3 \phi _4 \phi _5^2+6 \phi _5 \phi _6^2-3 \phi _3^2 \phi _4-6 \phi _4^2 \phi _5-9 \phi _5^2 \phi _6-3 \phi _3^2 \phi _7    )   \\
&\quad -2n_F    (\phi _4^3-3 \phi _3 \phi _4^2+3 \phi _6 \phi _4^2+3 \phi _3^2 \phi _4+3 \phi _6^2 \phi _4-6 \phi _5 \phi _6 \phi _4+\phi _6^3+\phi _7^3-{3\phi_1\phi_2^2-3\phi_2\phi_3^2-6\phi_5\phi_6^2}\\
& \quad+3  (-\phi _3+\phi _4+\phi _6    ) \phi _7^2+3 \phi _1^2 \phi _2+3 \phi _2^2 \phi _3+6 \phi _5^2 \phi _6+3    (   (\phi _3-\phi _4    ){}^2+\phi _6^2-2 \phi _4 \phi _6    ) \phi _7    )\\[1ex]
\end{aligned}
$\\
\hline
$ \begin{aligned}
6{\cal F}_6(\phi)  &=2  (1-n_A)   (4 \phi _1^3+3 \phi _2 \phi _1^2-6 \phi _2^2 \phi _1+4 \phi _2^3+4 \phi _3^3+4 \phi _4^3+4 \phi _5^3+4 \phi _6^3\\
&\quad+4 \phi _7^3-3 \phi _2 \phi _3^2+3 \phi _4 \phi _5^2+6 \phi _5 \phi _6^2-3 \phi _3^2 \phi _4-6 \phi _4^2 \phi _5-9 \phi _5^2 \phi _6-3 \phi _3^2 \phi _7   )  \\  
 & \quad -2   n_F  (\phi _5^3-3 \phi _4 \phi _5^2+3 \phi _6 \phi _5^2+3 \phi _4^2 \phi _5-3 \phi _6^2 \phi _5+2 \phi _7^3-3 \phi _1 \phi _2^2\\
 & \quad -3 \phi _2 \phi _3^2-3 \phi _3 \phi _4^2+3     (\phi _5-\phi _3   ) \phi _7^2+3 \phi _1^2 \phi _2+3 \phi _2^2 \phi _3\\
& \quad+3 \phi _3^2 \phi _4+3     (\phi _3^2-2 \phi _4 \phi _3+2 \phi _4^2+\phi _5^2+2 \phi _6^2-2 \phi _4 \phi _5-2 \phi _5 \phi _6   ) \phi _7   ) \\[1ex]
\end{aligned}
$
\\
\hline
$ \begin{aligned}
6 {\cal F}_7(\phi)  & = 2(1-n_A)  (4 \phi _1^3+3 \phi _2 \phi _1^2-6 \phi _2^2 \phi _1+4 \phi _2^3+4 \phi _3^3+4 \phi _4^3+4 \phi _5^3+4 \phi _6^3\\
& \quad +4 \phi _7^3-3 \phi _2 \phi _3^2+3 \phi _4 \phi _5^2+6 \phi _5 \phi _6^2-3 \phi _3^2 \phi _4-6 \phi _4^2 \phi _5-9 \phi _5^2 \phi _6-3 \phi _3^2 \phi _7  ) \\
&\quad 
-6 n_F (\phi _7^3-\phi _3 \phi _7^2+ (\phi _3^2-2 \phi _4 \phi _3+2  (\phi _4^2-\phi _5 \phi _4+\phi _5^2+\phi _6^2-\phi _5 \phi _6  )  ) \phi _7\\
&\quad -\phi _1 \phi _2^2-\phi _2 \phi _3^2-\phi _3 \phi _4^2-\phi _4 \phi _5^2-\phi _5 \phi _6^2+\phi _1^2 \phi _2 +\phi _2^2 \phi _3+\phi _3^2 \phi _4+\phi _4^2 \phi _5+\phi _5^2 \phi _6  )\\[1ex]
\end{aligned}
$
\\
\hline
$
 \begin{aligned}
6{\cal F}_8(\phi) &= 
2   (1-n_A)(4 \phi _1^3+3 \phi _2 \phi _1^2-6 \phi _2^2 \phi _1+4 \phi _2^3+4 \phi _3^3+4 \phi _4^3+4 \phi _5^3+4 \phi _6^3+4 \phi _7^3\\
&\quad -3 \phi _2 \phi _3^2+3 \phi _4 \phi _5^2+6 \phi _5 \phi _6^2-3 \phi _3^2 \phi _4-6 \phi _4^2 \phi _5-9 \phi _5^2 \phi _6-3 \phi _3^2 \phi _7  ) 
\\
&\quad -2 n_F   (2 \phi _4^3-3 \phi _3 \phi _4^2+3 \phi _3^2 \phi _4+6 \phi _6^2 \phi _4-6 \phi _5 \phi _6 \phi _4-3 \phi _1 \phi _2^2-3 \phi _2 \phi _3^2\\
&\quad -6 \phi _5 \phi _6^2-3    (\phi _3-2 \phi _4  ) \phi _7^2+3 \phi _1^2 \phi _2+3 \phi _2^2 \phi _3+6 \phi _5^2 \phi _6+3 \phi _3    (\phi _3-2 \phi _4  ) \phi _7  ) \\[1ex]
\end{aligned}
$
\\
\hline
\end{tabular}}
\end{center}
\caption{Prepotential in the eight chambers of the Coulomb branch of E$_7$ with $n_A$ hypermultiplet charged in the  adjoint representation ($\bf{133}$) and 
$n_F$ hypermultiplets  charged in the fundamental representation ($\mathbf{56}$). \label{Table:Prepotential}}
\end{table}

\clearpage

\section{Preliminaries}

\subsection{$G$-models and Coulomb phases} 
Let $\varphi: X\longrightarrow B$ be an elliptic fibration defined over the complex numbers, so that $\varphi$ is a proper morphism between complex quasi-projective varieties. We will assume that the base $B$ of the fibration is a smooth complex variety and denote the discriminant locus by $\Delta$. Under mild assumptions, $\Delta$ is a Cartier divisor. 

Inspired by  the physics of F-theory, we attach to  the elliptic fibration $\varphi:Y\longrightarrow B$  a unique reductive complex Lie group $G$ and a representation $\mathbf{R}$ of its Lie algebra $\mathfrak{g}=\text{Lie}(G)$. 
The  Lie algebra $\mathfrak{g}$ is uniquely defined by the dual graphs of the singular fibers of $\varphi$ over the generic points of its discriminant locus $\Delta$, and $\varphi$ is referred to as a {\em $G$-model}.\footnote{If $\Delta_i$ are the irreducible components of the reduced discriminant and  $p_i$ is the generic point of $\Delta_i$, the Langlands dual of the dual graph of the fiber over $p_i$ has an affine Dynkin type $\widetilde{\mathfrak{g}}_i$ and 
the Lie algebra $\mathfrak{g}$ is the direct sum $\mathfrak{g}=\bigoplus_i\mathfrak{g}_i$. By definition, the Kodaira type associated to a component $\Delta_i$ is the type of the  geometric fiber over $p_i$. 
 The dual graph of a Kodaira fiber is always an affine Dynkin diagram of type ADE. 
 We assume that  at least one component $\Delta_i$ of the reduced discriminant has a reducible  fiber over its generic point, that is, a Kodaira fiber of type different from type II and type I$_1$.}  
 
The irreducible components of the singular fibers over codimension-two points of the base determine a finite set of weights belonging to the representation $\mathbf{R}$. 
Determining $\mathbf{R}$ from these weights is particularly straightforward when $\mathbf{R}$ is a (quasi)-minuscule representation, as in that case all the (nonzero) weights of $\mathbf{R}$ are then by definition in the same Weyl orbit. 
The group $G$ is such that its first homotopy group is isomorphic to the Mordell--Weil group of $\varphi$, and $\mathbf{R}$ is a representation  of $G$. 
Several aspects of the  geometry of the elliptic fibration $\varphi$ are controlled by the triple  $(\mathfrak{g},G,\mathbf{R})$. 

 M-theory compactified on a Calabi--Yau threefold given by a $G$-model results in an $\mathcal{N}=1$  five-dimensional  gauged supergravity theory with gauge group $G$.   In the language of five-dimensional supersymmetric gauge theories with eight supercharges, each distinct crepant resolution of the Weierstrass model corresponds to a different chamber of the Coulomb branch of the theory. 
  String dualities suggest that the graph of flops between distinct minimal models is isomorphic to the adjacency graph of the chambers of the hyperplane arrangement  I($\mathfrak{g},\mathbf{R}$) whose  hyperplanes are the kernels of the weights of the representation $\mathbf{R}$ restricted to the dual fundamental Weyl chamber of $\mathfrak{g}$.
  This picture  requires that the elliptic fibration is a  Calabi-Yau threefold, but we expect that the hyperplane arrangement is valid in a larger setup than the Calabi-Yau threefold case and even relevant to study  varieties more general than elliptic fibrations.  For example, it can be extended to the case of
 $\mathbb{Q}$-factorial terminal singularities that are partial resolutions of 
  threefolds with cDV singularities.

  \subsection{Geography of minimal models: decomposition of the relative movable cone into relative nef-cones.} 

  The description of the Coulomb branch of a five-dimensional gauge theory with gauge algebra $\mathfrak{g}$ and a representation $\mathbf{R}$ in terms of a hyperplane arrangement  I($\mathfrak{g},\mathbf{R}$) is strikingly similar to the point of view of birational geometers who rely on several cones to describe the birational geometry of projective varieties in the same birational class. In particular,  we refer to   Section 12.2 of Matsuki~\cite{Matsuki} and to Kawamata's seminal paper~\cite{Kawamata} for a description of the relative movable cone and its decomposition theorem into relative nef-cones.

Let $\varphi: Y\longrightarrow B$ be an elliptic fibration. 
Let $f:Y\rightarrow W$ be the birational map to its Weierstrass model. 
Then $f$ contracts all fibral divisors not touching the zero section of the fibration. 
In the case of an elliptic surface, that is enough to see that the Weierstrass model is the canonical model of $Y$. 
 When an elliptic fibration has a trivial Mordell-Weil group, the fibral divisors $D_m$ and the zero  section generate the relative N\'eron-Severi cone  N$^1(X/B)$. 
 
 Given  a crepant resolution $f:X\longrightarrow W$ of a Weierstrass model, the relative  Picard number $\rho(X/W)$ gives the rank of the gauge group and  N$^1(X/W)$  
is generated by the fibral divisors $D_m$ not touching the zero section. 
Each relative one-cycle $\varpi(C)$ defines a map N$^1(X/W)$ $\rightarrow\mathbb{Z}^{\rho(X/W)}$ via the intersection.  
The geometric weights of a relative 1-cycle $C$ is a vector $\varpi(C)\in \mathbb{Z}^{\rho(X/W)}$ given by the negative of the intersection numbers with the fibral divisors not touching the section of the elliptic fibration \cite{Aspinwall:1996nk,G2}: 
$$\varpi_m(C)=-\int_X {D}_i \cdot C,$$
where $D_i$ are the fibral divisors not touching the section of the elliptic fibration. 
Each relative one-cycle defines a hyperplane in N$^1(X/W)$: 
$$
\varpi^\bot(C)=\{ D\in \text{N}^1(X/W) |    \int_X (D\cdot C)=0\}.  
$$
Each divisor $D_i$ is a fibration over $S$ and we denote its generic fiber by $C_i$.  
The curves $C_i$ generate  an open convex cone in $\text{N}^1(X/W)$. 
We denote the open dual cone in  $\text{N}^1(X/W)$ by $\mathfrak{D}$: 
$$
\mathfrak{D}=\{ D \in \text{N}^1(X/W) |\quad \forall i  \int_X (D \cdot C_i)\geq 0 \}=\bigcap_{i=1}^{\rho(X/W)}  (C_i)_{\geq 0}.
$$
This cone corresponds to the dual Weyl chamber of $\mathfrak{g}$. Geometrically, we expect $\mathfrak{D}$ to be the cone of relative movable divisors $\overline{\text{Mov}}$($X/W$).

A pseudo-isomorphism is a birational map that is an isomorphism in codimension-one. 
A small $\mathbb{Q}$-factorial modification (SQM) is a pseudo-isomorphism between $\mathbb{Q}$-factorial projective varieties.  
     Given  a SQM between  two varieties $X_1$ and $X_2$, we get naturally an associated identifcation of the of the N\'eron-Severi groups N$^1(X_1)$ and N$^1(X_2)$ via pullback and pushforward.

 Any nef-divisor is movable and the partition theorem implies that  any movable divisor becomes nef after a finite number of flops. 
For any two $\mathbb{Q}$-factorial varieties Y$_1$ and Y$_2$ related by a birational map that is an isomorphism in codimension-one, the  cones N$^1(Y_1/W)$ and N$^1(Y_2/W)$ can be canonically identified by pushforward and pullback. We also identify in this way some other important subcones  of N$^1(Y/W)$ and N$^1(Y/W)$ such as the relative movable cones Mov($Y/W$), the ample cones Amp($Y/W$), and the big cones Big($Y/W$)
$$
\overline{\text{Amp}}(Y/S)\subset \overline{\text{Mov}}(Y/S)\subset \overline{\text{Big}}(Y/S)\subset \text{N}^1(Y/S).
$$
The closure of the ample cone is the nef-cone. 
In  ideal cases, we expect that the relative movable cone decomposes as a union of relative nef-cones of all the minimal models in the same birational orbit: 
$$
\overline{\text{Mov}}(Y/S)=\bigsqcup_i    \overline{\text{Amp}}(Y_i/S),
$$
where the union is over all minimal models in the same birational class of $Y$ and the interiors of the  cones $\overline{\text{Amp}}(Y_i/S)$ are disjoint.

 Since the pullback of an ample divisor is movable, the identification of N$^1(X_1)$ and N$^1(X_2)$ also embeds the ample cone of $X_2$  onto a subcone of the movable cone of $X_1$. The partition theorem states that  if we fix a given minimal model $X$, the nef-cones of the other minimal models $X_i$  embedded into subcones of the movable cone of $X$ will provide a partition of the movable cone of $X$.

The hyperplane arrangement   I($\mathfrak{g},\mathbf{R}$) describes the decomposition  of the closed  relative movable cone  into relative nef-cones corresponding to each individual distinct crepant resolution. 
\begin{center}
\begin{tabular}{|c|c|c|}
\hline 
Elliptic fibration & Hyperplane arrangement & Cones \\
\hline 
Weierstrass model & Closed dual Weyl chamber  of $\mathfrak{g}$ & $\overline{\text{Mov}}(Y/S)$  \\
\hline 
Crepant resolution & Chambers of  I($\mathfrak{g},\mathbf{R}$) & $\overline{\text{Amp}}(Y_i/S)$  \\
\hline 
\end{tabular}
\end{center}

  \section{The hyperplane arrangement I($\text{E}_7, \mathbf{56}$) \label{sec:IE756}}
In this section, we review the construction of the root systems E$_7$, E$_8$, and the fundamental representation $\bf{56}$ of E$_7$. 
We then study the chamber structure of the hyperplane arrangement I($\text{E}_7, \bf{56}$). 

\subsection{E$_7$  root system and Dynkin diagram}
The Lie algebra of type E$_7$ has dimension $133$ and rank $7$: it has $126$ distinct roots and its Cartan subalgebra has dimension 7. 
Its Weyl group has order $2^{10}\cdot 3^4 \cdot 5 \cdot 7$ \cite[Plate VI]{Bourbaki.GLA46}.
The Coxeter number of the Lie algebra of type E$_7$ is $18$. The determinant of the Cartan matrix of the Lie algebra of type E$_7$ is $2$. Hence, the fundamental group of the root system of type E$_7$ and the quotient of the weight lattice modulo the root lattice are both isomorphic to $\mathbb{Z}/2\mathbb{Z}$. 

A choice of positive simple roots is
$$
\begin{array}{l}
\alpha_1= \frac{1}{2} (1, -1, -1, -1, -1, -1, -1,1),\\
\alpha_2=(-1, 1, 0 , 0 , 0, 0, 0, 0), \quad
\alpha_3=(0, -1, 1, 0, 0, 0, 0,0), \quad 
\alpha_4=(0, 0,-1, 1, 0, 0, 0, 0),\\ 
\alpha_5=(0, 0,0, -1, 1, 0, 0, 0), \quad
\alpha_6=(0, 0,0,0,  -1, 1, 0, 0), \quad
\alpha_7=(1,1,\   0,\   0, 0, 0, 0, 0).
\end{array}
$$
The Cartan matrix of E$_7$ is then 
\begin{equation}
\begin{array}{c}
\alpha_1 \\
\alpha_2 \\
\alpha_3 \\
\alpha_4 \\
\alpha_5 \\
\alpha_6 \\
\alpha_7 
\end{array}
\left(
\begin{array}{ccccccc}
 2 & -1 & 0 & 0 & 0 & 0 & 0 \\
 -1 & 2 & -1 & 0 & 0 & 0 & 0 \\
 0 & -1 & 2 & -1 & 0 & 0 & -1 \\
 0 & 0 & -1 & 2 & -1 & 0 & 0 \\
 0 & 0 & 0 & -1 & 2 & -1 & 0 \\
 0 & 0 & 0 & 0 & -1 & 2 & 0 \\
 0 & 0 & -1 & 0 & 0 & 0 & 2 \\
\end{array}
\right)
\end{equation}
The $i$th row of this Cartan matrix gives the coordinates of the simple root $\alpha_i$ in the basis of simple fundamental weights. 
In Bourbaki's tables,  our simple roots ($\alpha_1,\alpha_2,\alpha_3,\alpha_4,\alpha_5,\alpha_6,\alpha_7)$ are 
denoted ($\alpha_1,\alpha_3,\alpha_4,\alpha_5,\alpha_6,\alpha_7,\alpha_2)$, respectively.

\begin{figure}[htb]
\begin{center}
\scalebox{.9}{
\begin{tikzpicture}
				% E7
				\node[draw,circle,thick,scale=1,fill=black,label=below:{\scalebox{1.2}{ $\alpha_0$}}] (0) at (0,0){$1$};
				\node[draw,circle,thick,scale=1,label=below:{\scalebox{1.2}{$\alpha_1$}}] (1) at (1.2,0){$2$};
				\node[draw,circle,thick,scale=1,label=below:{\scalebox{1.2}{$\alpha_2$}}] (2) at (2.4,0){$3$};
				\node[draw,circle,thick,scale=1,label=below:{\scalebox{1.2}{$\alpha_3$}}] (3) at (3.6,0){$4$};
				\node[draw,circle,thick,scale=1,label=below:{\scalebox{1.2}{$\alpha_4$}}] (4) at (4.8,0){$3$};
				\node[draw,circle,thick,scale=1,label=below:{\scalebox{1.2}{$\alpha_5$}}] (5) at (6,0){$2$};
				\node[draw,circle,thick,scale=1,label=below:{\scalebox{1.2}{$\alpha_6$}}] (6) at (7.2,0){$1$};
				\node[draw,circle,thick,scale=1, label=above:{\scalebox{1.2}{$\alpha_7$}}] (7) at (3.6,1.2){$2$};
				\draw[thick] (0)--(1)--(2)--(3)--(4)--(5)--(6);
				\draw[thick]  (3)--(7);
					\end{tikzpicture}}
					\caption{Affine  Dynkin diagram of type $\widetilde{\text{E}}_7$. 
					Removing the black node gives the Dynkin diagram of type E$_7$. 
					The numbers inside the nodes are the multiplicities of the Kodaira fiber of type III$^*$ and  the Dynkin labels of the highest root.  Thus,  the Coxeter number of E$_7$ is $18$. \label{Fig:E7}
					The root $\alpha_1$ is the highest weight of the adjoint representation while $\alpha_6$ is the highest weight of the fundamental represenation $\mathbf{56}$. 
					}
\end{center}
\end{figure}
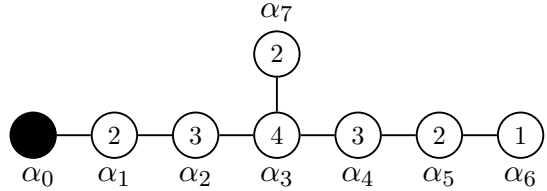

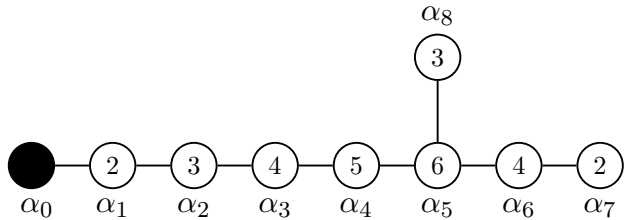
\begin{figure}[htb]
\begin{center}
\scalebox{.9}{
\def\arraystretch{1.5}% 
\begin{tikzpicture}
				% E8
				\node[draw,circle,thick,scale=1,fill=black,label=below:{\scalebox{1.2}{ $\alpha_0$}}] (0) at (0,0){$1$};
				\node[draw,circle,thick,scale=1,label=below:{\scalebox{1.2}{$\alpha_1$}}] (1) at (1.2,0){$2$};
				\node[draw,circle,thick,scale=1,label=below:{\scalebox{1.2}{$\alpha_2$}}] (2) at (2.4,0){$3$};
				\node[draw,circle,thick,scale=1,label=below:{\scalebox{1.2}{$\alpha_3$}}] (3) at (3.6,0){$4$};
				\node[draw,circle,thick,scale=1,label=below:{\scalebox{1.2}{$\alpha_4$}}] (4) at (4.8,0){$5$};
				\node[draw,circle,thick,scale=1,label=below:{\scalebox{1.2}{$\alpha_5$}}] (5) at (6,0){$6$};
				\node[draw,circle,thick,scale=1,label=below:{\scalebox{1.2}{$\alpha_6$}}] (6) at (7.2,0){$4$};
				\node[draw,circle,thick,scale=1, label=below:{\scalebox{1.2}{$\alpha_7$}}] (7) at (8.4,0){$2$};
				\node[draw,circle,thick,scale=1, label=above:{\scalebox{1.2}{$\alpha_8$}}] (8) at (6,1.6){$3$};
				\draw[thick] (0)--(1)--(2)--(3)--(4)--(5)--(6)--(7);
				\draw[thick]  (5)--(8);
					\end{tikzpicture}}
					\caption{Affine  Dynkin diagram of type $\widetilde{\text{E}}_8$. 
					Removing the black node gives the Dynkin diagram of type E$_8$. 
					The number in the nodes are the multiplicities of the Kodaira fiber of type II$^*$. 
					\label{Fig:E8}
					}
\end{center}
\end{figure}

\subsection{Root system of types E$_7$ and E$_8$, and representation $\mathbf{56}$ of E$_7$.}

We now give a quick description of the root system of type E$_7$ and  the weight system  $\mathbf{56}$ of E$_7$ in terms of the root system of type E$_8$. We follow  Borcherds's lecture notes on Lie groups~\cite{Borcherds}. 

The smallest non-trivial representation of E$_7$ is of dimension 56 and often called the fundamental representation.
The representation $\mathbf{56}$ of E$_7$ is  minuscule, self-dual, and pseudo-real. 
Its highest weight is the simple root $\alpha_6$  (see Figure \ref{Fig:E7}). 
The roots of E$_8$ form the unique  eight dimensional  even unimodular lattice. 
 The roots of E$_8$  are the vertices of the Gosset polytope ($\mathbf{4_{21}}$ in Coxeter's notation).  The roots of $\mathfrak{e}_7$ are the vertices of the polytope $\mathbf{2_{31}}$. 
The weights of $\mathbf{56}$ are the vertices of a  Delaunay polytope $\mathbf{3_{21}}$  also  called  the Hesse polytope by Conway and Sloane.

 The Lie algebra of type E$_8$ has the following decomposition under the maximal subalgebra E$_7\oplus A_1$:
\begin{equation}
\mathfrak{e}_8= (\mathfrak{e}_7\otimes \bf{1}) \oplus (\bf{1}\otimes \bf{su}_2)\oplus (\bf{56} \otimes \bf{2}). 
\end{equation}
The roots of E$_8$ are  the following $240$ vectors of $\mathbb{R}^8$ all located on a seven-dimensional  sphere of radius $\sqrt{2}$:
\begin{align}\begin{array}{ll}
 \frac{1}{2}(\pm 1, \pm 1,\pm 1, \pm 1,\pm 1, \pm 1,\pm 1, \pm 1) & (\text{with an even number of minus signs}) \\
(\pm 1, \pm 1, 0,0,0,0,0,0) & (\text{and permutations thereof}) 
\end{array}\end{align}
By definition, the root system  E$_7$ consists of the subsystem of  roots of  E$_8$ perpendicular to a chosen root $s$ of E$_8$.
 The Weyl group $W(E_8)$ preserves the scalar product $(r_1,r_2)$ between roots. 
All the roots of E$_8$ can be organized by their scalar products with respect to the chosen root $s$ of E$_8$. 
Since the root system E$_8$ is crystallographic and simply-laced, the scalar product of two roots can only take one of the following five values:  
$-2$, $-1$, $0$, $+1$, $+ 2$.  
Then, the following facts are directly proven by using the scalar product:
\begin{itemize}
\item There is a unique  root  $r$ such that $(s,r)=2$ (resp. $-2$),  namely $ s$  (resp. $-s$).  
\item There are $126$ roots perpendicular to $s$, they form  an  E$_7$ root system. 
\item There are 56 roots $r$ of $E_8$ with scalar product $(s,r)=1$ (resp. $-1$), we call them $\bf{56}_+$ (resp. $\bf{56}_-$). Each of $\mathbf{56}_\pm$  form  a weight system of the  irreducible representation $\mathbf{56}$  of E$_7$.
\item Translation by $-s$ gives an isometry  between the weight systems  $\mathbf{56}_+$ and $\mathbf{56}_-$ corresponding to the hyperplane reflection  induced by the root $s$.  
\item The involution  $r\to -r$ of E$_8$  exchanges $\mathbf{56}_+$ and $\mathbf{56}_-$. 
\end{itemize}

Using the scalar product of E$_8$, the weight system  $\mathbf{56}_+$ is composed of the closest roots to $s$. They are all on a sphere of radius $\sqrt{2}$ centered at the affine point $P_s$ defined by the vector $s$. 
The roots of E$_7$ are all on the plane perpendicular to the vector $s$ on a sphere of radius $2$ centered at the affine  point $P_s$ corresponding to the vector $s$. 
The vectors of the weight system $\mathbf{56}_-$ are on a sphere of radius $\sqrt{6}$ centered at the point $P_s$. 
 The root $-s$ is the further away at a distance $\sqrt{8}$.
 The root system E$_7$ defined by $(r,s)=0$ and the two representations $\mathbf{56}$ are on parallel hyperplanes.

Consider the root $s=(0,0,0,0,0,0,1,1)$ of E$_8$.
We can write an E$_7$ root system by listing the roots of E$_8$ perpendicular to $s$:
\begin{align}
\begin{array}{ll}
\pm \frac{1}{2} (\pm 1, \pm 1, \pm 1, \pm 1, \pm 1, \pm 1, 1,-1)  & (\text{with even number of minus signs}) \\
 (\pm 1, \pm 1, 0,0,0,0,0,0) &  (\text{and permutations of the first six coordinates thereof}) \\
  \pm (0,0,0,0,0,0,1,-1) & 
 \end{array}
\end{align}
We can also determine the weights of the representation $\mathbf{56}$ corresponding to $\bf{56}_+$ by enumerating the roots of E$_8$ with scalar product $+1$ with $s$:
\begin{align}\begin{array}{ll}
& \frac{1}{2} (\pm 1, \pm 1, \pm 1, \pm 1, \pm 1, \pm 1, 1,1)  \quad (\text{with even number of minus signs}) \\
 &(\pm 1, 0, 0,0,0,0,1,0)\quad \quad (\text{and permutations of the first six coordinates thereof}) \\
&(\pm 1, 0, 0,0,0,0,0,1)\quad \quad (\text{and permutations of the first six coordinates thereof})
 \end{array}\end{align}
The opposite of these vectors form the weight system $\bf{56}_-$. 
The representation $\bf{56}_\pm$ is invariant under the group $(\mathbb{Z}/2\mathbb{Z})^2$ generated by the involutions $r\mapsto \pm s-r$ and the reflection $\sigma_s: r\to -r + (s,r) s$. 

In the basis of fundamental weights, the weights of the representation $\bf{56}$ are listed with the corresponding Hasse diagram is given in Figure \ref{Fig:Hasse56}. 

\clearpage

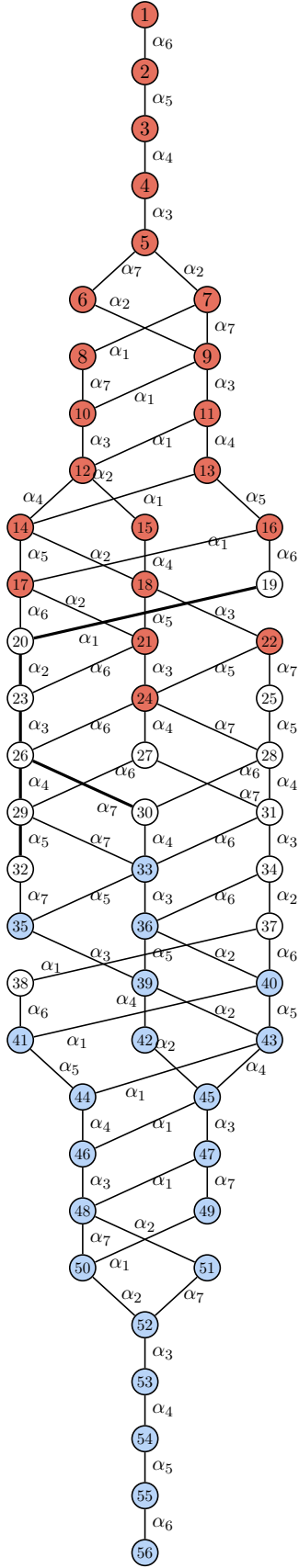
\begin{figure}[htb]
\begin{minipage}[c]{0.3\textwidth}
\scalebox{.75}{
\begin{tikzpicture}[ 
x=1.2cm,y=1.1cm,rotate=0,transform shape, color=black]
\tikzstyle{weight}=[circle,thick,draw,minimum size=5mm,inner sep=1pt];
\tikzstyle{root}=[minimum size=0.5cm];
\tikzstyle{sign}=[circle,thick,draw,minimum size=5mm,inner sep=1pt];
\tikzstyle{positive}=[circle,thick,draw,minimum size=5mm,inner sep=1pt,fill=myred];
\tikzstyle{negative}=[circle,thick,draw,minimum size=5mm,inner sep=1pt,fill=myblue];
				% E8
				\node[positive] (1) at (0,0) {$1$};
			         \node[positive] (2) at (0,-1) {$2$};
			         \node[positive] (3) at (0,-2) {$3$};
			         \node[positive] (4) at (0,-3) {$4$};
				\node[positive] (5) at (0,-4) {$5$};
			         \node[positive] (6) at (-1,-5) {$6$} ; \node[positive] (7) at (1,-5) {$7$};
			         			         \node[positive] (8) at (-1,-6) {$8$} ; \node[positive] (9) at (1,-6) {$9$};
						           \node[positive] (10) at (-1,-7) {\scalebox{.8}{$10$}} ; \node[positive] (11) at (1,-7) {\scalebox{.8}{$11$}};
			           \node[positive] (12) at (-1,-8) {\scalebox{.8}{$12$}} ; \node[positive] (13) at (1,-8) {\scalebox{.8}{$13$}};
			           \node[positive] (14) at (-2,-9) {\scalebox{.8}{$14$}} ; \node[positive] (15) at (0,-9) {\scalebox{.8}{$15$}}; ; \node[positive] (16) at (2,-9) {\scalebox{.8}{$16$}};
  \node[positive] (17) at (-2,-10) {\scalebox{.8}{$17$}} ; \node[positive] (18) at (0,-10) {\scalebox{.8}{$18$}}; ; \node[sign] (19) at (2,-10) {\scalebox{.8}{$19$}};						         
  \node[sign] (20) at (-2,-11) {\scalebox{.8}{$20$}} ; \node[positive] (21) at (0,-11) {\scalebox{.8}{$21$}}; ; \node[positive] (22) at (2,-11) {\scalebox{.8}{$22$}};						         
    \node[sign] (23) at (-2,-12) {\scalebox{.8}{$23$}} ; \node[positive] (24) at (0,-12) {\scalebox{.8}{$24$}}; ; \node[sign] (25) at (2,-12) {\scalebox{.8}{$25$}};					         
      \node[sign] (26) at (-2,-13) {\scalebox{.8}{$26$}} ; \node[sign] (27) at (0,-13) {\scalebox{.8}{$27$}}; ; \node[sign] (28) at (2,-13) {\scalebox{.8}{$28$}};						         
        \node[sign] (29) at (-2,-14) {\scalebox{.8}{$29$}} ; \node[sign] (30) at (0,-14) {\scalebox{.8}{$30$}} ; \node[sign] (31) at (2,-14) {\scalebox{.8}{$31$}};						                
          \node[sign] (32) at (-2,-15) {\scalebox{.8}{$32$}} ; \node[negative] (33) at (0,-15) {\scalebox{.8}{$33$}}; \node[sign] (34) at (2,-15) {\scalebox{.8}{$34$}};						            
				         
            \node[negative] (35) at (-2,-16) {\scalebox{.8}{$35$}} ; \node[negative] (36) at (0,-16) {\scalebox{.8}{$36$}}; \node[sign] (37) at (2,-16) {\scalebox{.8}{$37$}};	
              \node[sign] (38) at (-2,-17) {\scalebox{.8}{$38$}} ; \node[negative] (39) at (0,-17) {\scalebox{.8}{$39$}} ; \node[negative] (40) at (2,-17) {\scalebox{.8}{$40$}};	
                \node[negative] (41) at (-2,-18) {\scalebox{.8}{$41$}} ; \node[negative] (42) at (0,-18) {\scalebox{.8}{$42$}} ; \node[negative] (43) at (2,-18) {\scalebox{.8}{$43$}};

         \node[negative] (56) at (0,-27) {\scalebox{.8}{$56$}}; 
			         \node[negative] (55) at (0,-26) {\scalebox{.8}{$55$}}; 
			         \node[negative] (54) at (0,-25) {\scalebox{.8}{$54$}}; 
			         \node[negative] (53) at (0,-24) {\scalebox{.8}{$53$}}; 
				\node[negative] (52) at (0,-23) {\scalebox{.8}{$52$}}; 
			         \node[negative] (51) at (1,-22) {\scalebox{.8}{$51$}};  \node[negative] (50) at (-1,-22) {\scalebox{.8}{$50$}}; ;
			         			         \node[negative] (49) at (1,-21){\scalebox{.8}{$49$}};  \node[negative] (48) at (-1,-21)  {\scalebox{.8}{$48$}}; 
						           \node[negative] (47) at (1,-20) {\scalebox{.8}{$47$}} ; \node[negative] (46) at (-1,-20) {\scalebox{.8}{$46$}};
						       \node[negative] (44) at (-1,-19) {\scalebox{.8}{$44$}};      \node[negative] (45) at (1,-19) {\scalebox{.8}{$45$}} ;

\draw[thick]  (1)--   node[root,right] {$\alpha_{6}$}  (2);  
\draw[thick] (2)-- node[root,right] {$\alpha_{5}$} (3);
\draw[thick] (3)-- node[root,right] {$\alpha_{4}$} (4);      

\draw[thick] (4)-- node[root,right] {$\alpha_{3}$} (5);

\draw[thick] (5)-- node[root,right] {$\alpha_{7}$} (6);

\draw[thick] (5)-- node[root,right] {$\alpha_{2}$} (7);

\draw[thick] (6)-- node[root,above, near start] {$\alpha_{2}$} (9);

\draw[thick] (7)-- node[root,below, near end] {$\alpha_{1}$} (8);

\draw[thick] (7)-- node[root,right] {$\alpha_{7}$} (9);

\draw[thick] (8)-- node[root,right] {$\alpha_{7}$} (10);

\draw[thick] (9)-- node[root,below] {$\alpha_{1}$} (10);

\draw[thick] (9)-- node[root,right] {$\alpha_{3}$} (11);

\draw[thick] (10)-- node[root,right] {$\alpha_{3}$} (12);

\draw[thick] (11)-- node[root,right] {$\alpha_{1}$} (12);

\draw[thick] (11)-- node[root,right] {$\alpha_{4}$} (13);

\draw[thick] (12)-- node[root,left] {$\alpha_{4}$} (14);

\draw[thick] (12)-- node[root,above, near start] {$\alpha_{2}$} (15);

\draw[thick] (13)-- node[root,below, near start] {$\alpha_{1}$} (14);

\draw[thick] (13)-- node[root,right] {$\alpha_{5}$} (16);

\draw[thick] (14)-- node[root,right] {$\alpha_{5}$} (17);

\draw[thick] (14)-- node[root,right] {$\alpha_{2}$} (18);

\draw[thick] (15)-- node[root,right,near end] {$\alpha_{4}$} (18);

\draw[thick] (16)-- node[root,right, near start] {$\alpha_{1}$} (17);

\draw[thick] (16)-- node[root,right] {$\alpha_{6}$} (19);

\draw[thick] (17)-- node[root,right] {$\alpha_{6}$} (20);

\draw[thick] (17)-- node[root,right, near start] {$\alpha_{2}$} (21);

\draw[thick] (18)-- node[root,right,near end] {$\alpha_{5}$} (21);

\draw[thick] (18)-- node[root,right] {$\alpha_{3}$} (22);

\draw[ultra thick] (19)-- node[root,below, near end] {$\alpha_{1}$} (20);

\draw[ultra thick] (20)-- node[root,right] {$\alpha_{2}$} (23);

\draw[thick] (21)-- node[root,right] {$\alpha_{6}$} (23);

\draw[thick] (21)-- node[root,right] {$\alpha_{3}$} (24);

\draw[thick] (22)-- node[root,right] {$\alpha_{5}$} (24);

\draw[thick] (22)-- node[root,right] {$\alpha_{7}$} (25);

\draw[ultra thick] (23)-- node[root,right] {$\alpha_{3}$} (26);

\draw[thick] (24)-- node[root,right] {$\alpha_{6}$} (26);

\draw[thick] (24)-- node[root,right] {$\alpha_{4}$} (27);

\draw[thick] (24)-- node[root,right] {$\alpha_{7}$} (28);

\draw[thick] (25)-- node[root,right] {$\alpha_{5}$} (28);

\draw[ultra thick] (26)-- node[root,right] {$\alpha_{4}$} (29);

\draw[ultra thick] (26)-- node[root,below, near end] {$\alpha_{7}$} (30);

\draw[thick] (27)-- node[root,right, near start] {$\alpha_{6}$} (29);

\draw[thick] (27)-- node[root,right, near end] {$\alpha_{7}$} (31);

\draw[thick] (28)-- node[root,right,near start] {$\alpha_{6}$} (30);
\draw[thick] (28)-- node[root,right] {$\alpha_{4}$} (31);
\draw[ultra thick] (29)-- node[root,right] {$\alpha_{5}$} (32);
\draw[thick] (29)-- node[root,right] {$\alpha_{7}$} (33);
\draw[thick] (30)-- node[root,right] {$\alpha_{4}$} (33);
\draw[thick] (31)-- node[root,right] {$\alpha_{6}$} (33);

\draw[thick] (31)-- node[root,right] {$\alpha_{3}$} (34);

\draw[thick] (32)-- node[root,right] {$\alpha_{7}$} (35);

\draw[thick] (33)-- node[root,right] {$\alpha_{5}$} (35);

\draw[thick] (33)-- node[root,right] {$\alpha_{3}$} (36);

\draw[thick] (34)-- node[root,right] {$\alpha_{6}$} (36);

\draw[thick] (34)-- node[root,right] {$\alpha_{2}$} (37);

\draw[thick] (35)-- node[root,right] {$\alpha_{3}$} (39);

\draw[thick] (36)-- node[root,right] {$\alpha_{5}$} (39);

\draw[thick] (36)-- node[root,right] {$\alpha_{2}$} (40);

\draw[thick] (37)-- node[root,above, xshift=-1.8 cm, yshift=- .5cm] {$\alpha_{1}$} (38);

\draw[thick] (37)-- node[root,right] {$\alpha_{6}$} (40);

\draw[thick] (38)-- node[root,right] {$\alpha_{6}$} (41);

\draw[thick] (39)-- node[root,left,yshift=.2cm] {$\alpha_{4}$} (42);

\draw[thick] (39)-- node[root,right] {$\alpha_{2}$} (43);

\draw[thick] (40)-- node[root,xshift=-1.3cm, yshift=-.6cm] {$\alpha_{1}$} (41);

\draw[thick] (40)-- node[root,right] {$\alpha_{5}$} (43);
\draw[thick] (41)-- node[root,right] {$\alpha_{5}$} (44);

\draw[thick] (42)-- node[root,above, near start] {$\alpha_{2}$} (45);

\draw[thick] (43)-- node[root,below,near end] {$\alpha_{1}$} (44);

\draw[thick] (43)-- node[root,right] {$\alpha_{4}$} (45);

\draw[thick] (44)-- node[root,right] {$\alpha_{4}$} (46);

\draw[thick] (45)-- node[root,right] {$\alpha_{1}$} (46);

\draw[thick] (45)-- node[root,right] {$\alpha_{3}$} (47);

\draw[thick] (46)-- node[root,right] {$\alpha_{3}$} (48);

\draw[thick] (47)-- node[root,right] {$\alpha_{1}$} (48);

\draw[thick] (47)-- node[root,right] {$\alpha_{7}$} (49);

\draw[thick] (48)-- node[root,right] {$\alpha_{7}$} (50);

\draw[thick] (48)-- node[root,above] {$\alpha_{2}$} (51);

\draw[thick] (49)-- node[root,below, near end] {$\alpha_{1}$} (50);

\draw[thick] (50)-- node[root,right] {$\alpha_{2}$} (52);

\draw[thick] (51)-- node[root,right] {$\alpha_{7}$} (52);

\draw[thick] (52)-- node[root,right] {$\alpha_{3}$} (53);

\draw[thick] (53)-- node[root,right] {$\alpha_{4}$} (54);

\draw[thick] (54)-- node[root,right] {$\alpha_{5}$} (55);

\draw[thick] (55)-- node[root,right] {$\alpha_{6}$} (56);

					\end{tikzpicture}
					}
					\end{minipage}
					\hfill
					\begin{minipage}[c]{0.7\textwidth}
					\begin{center}
					\scalebox{.75}{
					$
  \begin{array}{c}
 \varpi_{1}\\
\varpi_{2}\\
\varpi_{3}\\
\varpi_{4}\\
\varpi_{5}\\
\varpi_{6}\\
\varpi_{7}\\
\varpi_{8}\\
\varpi_{9}\\
\varpi_{10}\\
\varpi_{11}\\
\varpi_{12}\\
\varpi_{13}\\
\varpi_{14}\\
\varpi_{15}\\
\varpi_{16}\\
\varpi_{17}\\
\varpi_{18}\\
\varpi_{19}\\
\varpi_{20}\\
\varpi_{21}\\
\varpi_{22}\\
\varpi_{23}\\
\varpi_{24}\\
\varpi_{25}\\
\varpi_{26}\\
\varpi_{27}\\
\varpi_{28}
\end{array}
 \left[
\begin{array}{ccccccc}
 0 & 0 & 0 & 0 & 0 & 1 & 0 \\
 0 & 0 & 0 & 0 & 1 & -1 & 0 \\
 0 & 0 & 0 & 1 & -1 & 0 & 0 \\
 0 & 0 & 1 & -1 & 0 & 0 & 0 \\
 0 & 1 & -1 & 0 & 0 & 0 & 1 \\
 0 & 1 & 0 & 0 & 0 & 0 & -1 \\
 1 & -1 & 0 & 0 & 0 & 0 & 1 \\
 -1 & 0 & 0 & 0 & 0 & 0 & 1 \\
 1 & -1 & 1 & 0 & 0 & 0 & -1 \\
 -1 & 0 & 1 & 0 & 0 & 0 & -1 \\
 1 & 0 & -1 & 1 & 0 & 0 & 0 \\
 -1 & 1 & -1 & 1 & 0 & 0 & 0 \\
 1 & 0 & 0 & -1 & 1 & 0 & 0 \\
 -1 & 1 & 0 & -1 & 1 & 0 & 0 \\
 0 & -1 & 0 & 1 & 0 & 0 & 0 \\
 1 & 0 & 0 & 0 & -1 & 1 & 0 \\
 -1 & 1 & 0 & 0 & -1 & 1 & 0 \\
 0 & -1 & 1 & -1 & 1 & 0 & 0 \\
 1 & 0 & 0 & 0 & 0 & -1 & 0 \\
 -1 & 1 & 0 & 0 & 0 & -1 & 0 \\
 0 & -1 & 1 & 0 & -1 & 1 & 0 \\
 0 & 0 & -1 & 0 & 1 & 0 & 1 \\
 0 & -1 & 1 & 0 & 0 & -1 & 0 \\
 0 & 0 & -1 & 1 & -1 & 1 & 1 \\
 0 & 0 & 0 & 0 & 1 & 0 & -1 \\
 0 & 0 & -1 & 1 & 0 & -1 & 1 \\
 0 & 0 & 0 & -1 & 0 & 1 & 1 \\
 0 & 0 & 0 & 1 & -1 & 1 & -1 \\
\end{array}
\right]
\quad 
 \begin{array}{c}
 \varpi_{29}\\
\varpi_{30}\\
\varpi_{31}\\
\varpi_{32}\\
\varpi_{33}\\
\varpi_{34}\\
\varpi_{35}\\
\varpi_{36}\\
\varpi_{37}\\
\varpi_{38}\\
\varpi_{39}\\
\varpi_{40}\\
\varpi_{41}\\
\varpi_{42}\\
\varpi_{43}\\
\varpi_{44}\\
\varpi_{45}\\
\varpi_{46}\\
\varpi_{47}\\
\varpi_{48}\\
\varpi_{49}\\
\varpi_{50}\\
\varpi_{51}\\
\varpi_{52}\\
\varpi_{53}\\
\varpi_{54}\\
\varpi_{55}\\
\varpi_{56}\\
\end{array}
\left[
\begin{array}{ccccccc}
 0 & 0 & 0 & -1 & 1 & -1 & 1 \\
 0 & 0 & 0 & 1 & 0 & -1 & -1 \\
 0 & 0 & 1 & -1 & 0 & 1 & -1 \\
 0 & 0 & 0 & 0 & -1 & 0 & 1 \\
 0 & 0 & 1 & -1 & 1 & -1 & -1 \\
 0 & 1 & -1 & 0 & 0 & 1 & 0 \\
 0 & 0 & 1 & 0 & -1 & 0 & -1 \\
 0 & 1 & -1 & 0 & 1 & -1 & 0 \\
 1 & -1 & 0 & 0 & 0 & 1 & 0 \\
 -1 & 0 & 0 & 0 & 0 & 1 & 0 \\
 0 & 1 & -1 & 1 & -1 & 0 & 0 \\
 1 & -1 & 0 & 0 & 1 & -1 & 0 \\
 -1 & 0 & 0 & 0 & 1 & -1 & 0 \\
 0 & 1 & 0 & -1 & 0 & 0 & 0 \\
 1 & -1 & 0 & 1 & -1 & 0 & 0 \\
 -1 & 0 & 0 & 1 & -1 & 0 & 0 \\
 1 & -1 & 1 & -1 & 0 & 0 & 0 \\
 -1 & 0 & 1 & -1 & 0 & 0 & 0 \\
 1 & 0 & -1 & 0 & 0 & 0 & 1 \\
 -1 & 1 & -1 & 0 & 0 & 0 & 1 \\
 1 & 0 & 0 & 0 & 0 & 0 & -1 \\
 -1 & 1 & 0 & 0 & 0 & 0 & -1 \\
 0 & -1 & 0 & 0 & 0 & 0 & 1 \\
 0 & -1 & 1 & 0 & 0 & 0 & -1 \\
 0 & 0 & -1 & 1 & 0 & 0 & 0 \\
 0 & 0 & 0 & -1 & 1 & 0 & 0 \\
 0 & 0 & 0 & 0 & -1 & 1 & 0 \\
 0 & 0 & 0 & 0 & 0 & -1 & 0 \\
\end{array}
\right]
$
}
\end{center}
\caption{Hasse diagram for the weights of the representation $\bf{56}$ of E$_7$. The node $i$ represents the weight $\varpi_i$ given above and written in the basis of fundamental weights. 
A  root $\alpha$ between two nodes $i$ and $j$ (with $i<j$) indicates that $\varpi_j=\varpi_i -\alpha$.
 The white nodes are those corresponding to the weights used to define the sign vector that characterizes the chambers of the hyperplane arrangement I($\bf{56}$,~E$_7$). In a given chamber of  I($\bf{56}$, E$_7$), each white node takes a specific sign (see section \eqref{sec:IE756}). 
The red (resp. blue) nodes correspond to weights in the positive (resp. negative) conical hull of positive  roots. In particular, a white node corresponds to a weight $\varpi$ such that the hyperplane $\varpi^\bot$ intersects the  interior of the dual fundamental Weyl chamber while that is not the case for a weight corresponding to a blue or red node. 
}
\label{Fig:Hasse56}
\end{minipage}
						\end{figure}

\clearpage

\subsection{Chamber structure of the hyperplane I(E$_7$,$\mathbf{56}$) \label{sec:IE756}}

The following Theorem was given in \cite{Box} without a formal proof. We give a proof here using the language of sign vectors in the spirit of \cite{EJJN1,EJJN2}. 
\begin{thm}
The hyperplane arrangement  I$(E_7, \mathbf{56})$ has eight chambers. Each chamber is simplicial.  The adjacency graph of the chambers is isomorphic to the Dynkin diagram  of type E$_8$. 
\end{thm}

\begin{proof}
Since the minuscule representation of E$_7$ is self-dual, it is enough to consider only half of its weights to study the hyperplane arrangement I$(E_7, \mathbf{56})$. 
The weights that do not intersect the interior of the dual fundamental Weyl chamber are such that all their coefficients have the same sign when expressed in the basis of simple positive roots.  
After removing such weights, we are left (up to a sign) with seven weights, listed in Figure \ref{Fig:SV}.  They have an elegant structure as all can be derived from the smallest one by removing 
only one root, always different from $\alpha_6$.
~\\ 
~\\~\\
$	\begin{aligned}\begin{array}{llll}
\varpi_{19}&\quad (1, 1, 1, \frac{1}{2}, 0, \text{-}\frac{1}{2}, \frac{1}{2})\quad&\boxed{\ 1\ \ 0\ 0 \ 0 \ 0  \ $-1$ \ \ 0 } & \quad\quad\varpi_{19}\\
\varpi_{20}&\quad (0, 1, 1, \frac{1}{2}, 0, \text{-}\frac{1}{2}, \frac{1}{2}) \quad&\boxed{$-1$ \  \  1 \ 0 \ 0  \ 0 \ $-1$ \ \  0 } &\quad\quad \varpi_{19}-\alpha_1\\
\varpi_{23}&\quad (0, 0, 1, \frac{1}{2}, 0, \text{-}\frac{1}{2}, \frac{1}{2}) \quad&\boxed{\ 0 \  $-1$ \ 1 \ 0  \ 0 \ $-1$ \ \  0 } &\quad\quad \varpi_{19}-\alpha_1-\alpha_2\\
\varpi_{26}& \quad(0, 0, 0, \frac{1}{2}, 0, \text{-}\frac{1}{2}, \frac{1}{2}) \quad&\boxed{\  0 \ \  0 \ $-1$\ 1   \ 0 $\ -1$ \   1}&\quad\quad \varpi_{19}-\alpha_1-\alpha_2-\alpha_3\\
\varpi_{29}& \quad(0, 0, 0, \text{-}\frac{1}{2}, 0, \text{-}\frac{1}{2}, \frac{1}{2}) \quad&\boxed{\  0 \  \ 0 \ 0 \   $-1$  \ $1$ \ $-1$ \ $1$ }&\quad\quad \varpi_{19}-\alpha_1-\alpha_2-\alpha_3-\alpha_4\\
\varpi_{32}& \quad(0, 0, 0, \text{-}\frac{1}{2}, \text{-}1, \text{-}\frac{1}{2}, \frac{1}{2}) \quad&\boxed{\ 0 \   \ 0 \ 0 \ 0 \  $-1$\   0\   \ $1$ } &\quad\quad \varpi_{19}-\alpha_1-\alpha_2-\alpha_3-\alpha_4-\alpha_5\\
\varpi_{30}& \quad(0, 0, 0, \frac{1}{2}, 0,\text{-}\frac{1}{2},\text{-}\frac{1}{2}) \quad&\boxed{\ 0\  \ 0 \ 0 \ $1$  \ 0 \ $-1$ \ $-1$ }&\quad\quad \varpi_{19}-\alpha_1-\alpha_2-\alpha_3-\alpha_7 
\end{array}\end{aligned}
			$
		~	\\ 
		\\ 
\begin{figure}[b!]
\begin{center}
\begin{tikzpicture}[scale=.75]
			\tikzmath{\x1 = 1.75;};
				\node at (13.5,-7) {\scalebox{1}{
				$
				\begin{aligned}
\varpi_{19}=-\varpi_{38}&\quad \boxed{\ 1\ \ 0\ 0 \ 0 \ 0  \ $-1$ \ \ 0 } \\
\\
\varpi_{20}=-\varpi_{37}&\quad \boxed{$-1$ \  \  1 \ 0 \ 0  \ 0 \ $-1$ \ \  0 } && \\
\\
\varpi_{23}=-\varpi_{34}&\quad \boxed{\ 0 \  $-1$ \ 1 \ 0  \ 0 \ $-1$ \ \  0 }\\
\\
\varpi_{26}=-\varpi_{31}& \quad \boxed{\  0 \ \  0 \ $-1$\ 1   \ 0 \ $-1$ \   1}\\
\\
\varpi_{29}=-\varpi_{28}& \quad \boxed{\  0 \  \ 0 \ 0 \ $-1$  \ $1$ \ $-1$\  1 } & &\quad  \varpi_{30}=-\varpi_{27}& \  \boxed{\ 0\  \ 0 \ 0 \ 1  \ 0 \ $-1$ \ \  $-1$ }\\
\\
\varpi_{32}=-\varpi_{25}& \quad \boxed{\ 0 \   \ 0 \ 0 \ 0 \  $-1$\   0\   \ $1$ }
\end{aligned}
				$
							}};
				\node[draw,circle,thick,scale=1] (1) at (2,-1.2*\x1){$\varpi_{19}$};
				\node[draw,circle,thick,scale=1] (2) at (2,-2.4*\x1){$\varpi_{20}$};
				\node[draw,circle,thick,scale=1] (3) at (2,-3.6*\x1){$\varpi_{23}$};
				\node[draw,circle,thick,scale=1] (4) at (2,-4.8*\x1){$\varpi_{26}$};
				\node[draw,circle,thick,scale=1] (5) at (2,-6*\x1){$\varpi_{29}$};
				\node[draw,circle,thick,scale=1] (6) at (2,-7.2*\x1){$\varpi_{32}$};
				\node[draw,circle,thick,scale=1] (8) at (2.5*\x1,-6*\x1){$\varpi_{30}$};
				\draw[thick, ->] (1)--node[left] {$-\alpha_{1}$}(2);
				\draw[thick, ->] (2)--node[left] {$-\alpha_2$}(3);
				\draw[thick, ->](3)--node[left] {$-\alpha_3$}(4);
				\draw[thick, ->] (4)--node[left] {$-\alpha_4$}(5);
				\draw[thick, ->] (5)--node[left] {$-\alpha_5$}(6);
				\draw[thick,->] (4)--node[right] {$-\alpha_7$}(8);
					\end{tikzpicture}

\caption{Up to a sign, there are seven  weights of the representation $\mathbf{56}$ of E$_7$ such that each of them has a kernel   intersecting the interior of the dual fundamental Weyl chamber. 
The partial order of weights provide the Hasse diagram presented above and corresponding to a decorated Dynkin diagram of type E$_7$.  
We write $\varpi_i \xrightarrow{-\alpha_\ell} \varpi_j$ to indicate that  $\varpi_i-\alpha_\ell= \varpi_j$. 
The linear forms defined by these weights give the sign vector for the hyperplane arrangement I$(\text{E}_7, \mathbf{56})$.  
  \label{Fig:SV}}
\end{center}
\end{figure}
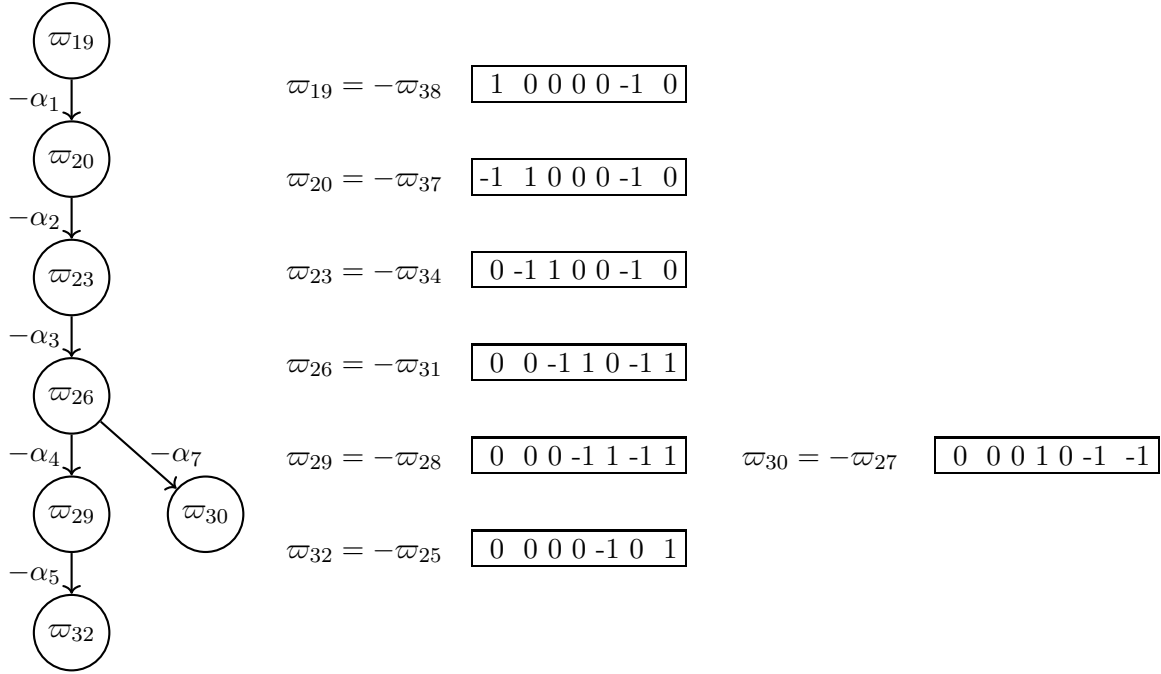\\
In the first column, the weights $\varpi$ are written in the basis of simple roots, while in the second column, they are written in the basis of simple fundamental weights. We denote by  $\phi$ an element of the coroot space. When writing $\phi$ in coordinates, we use the basis of simple coroots. Since, by definition, the simple coroots are dual to the fundamental weights,  the scalar product $\varpi\cdot \phi$ is defined using the identity matrix. 
We define a sign vector for the hyperplane arrangement I$(E_7, \mathbf{56})$ as follows:\footnote{Each  weight of the representation $\mathbf{56}$ has norm square $3/2$ and has scalar product $\pm 1/2$ with any other weight of $\mathbf{56}$. 
Our choice of signs for the entries of the sign vector is such that the  highest weight $\boxed{0\  0 \  0 \  0\  0\  1\  0}$ has a sign $(-1,-1,-1,-1,-1,-1,-1)$.} 
\begin{equation}
\phi\mapsto (\varpi_{19}\cdot \phi , \varpi_{20} \cdot \phi , \varpi_{23}\cdot \phi , \varpi_{26}\cdot \phi , \varpi_{29}\cdot \phi ,\varpi_{32}\cdot \phi,\varpi_{30}\cdot \phi ),
\end{equation}
which gives 
\begin{equation}\label{eq.SV}
\begin{aligned}
\sigma(\phi)= & (\phi_1-\phi_6, 
-\phi_1+\phi_2-\phi_6,
-\phi_2+\phi_3-\phi_6, \\
& \   -\phi_3+\phi_4-\phi_6+\phi_7,
-\phi_4+\phi_5-\phi_6+\phi_7,
-\phi_5+\phi_7,
\phi_4-\phi_6-\phi_7
). 
\end{aligned}
\end{equation}
{The open dual fundamental Weyl chamber of E$_7$ is defined as the cone  $(\alpha_i\cdot \phi)>0$ with  $i=\{1,2,3,4,5,6,7\}$.} 

Each chamber is uniquely defined by the signs of entries of the vector $\sigma(\phi)$ when evaluated at any interior point $\phi$ of the chamber. 
The weights have a partial order defined by adding positive simple roots: $\varpi_i\succ\varpi_j$ if $\varpi_i-\varpi_j$ is a nonnegative integer linear combination of positive roots. 
In particular, the partial order for the weights that are interior walls is (see Figure \ref{Fig:SV}): 
$$
\varpi_{19}\succ\varpi_{20}\succ\varpi_{23}\succ\varpi_{26}\succ\varpi_{29}\succ\varpi_{32}, \quad \varpi_{26}\succ\varpi_{30}.
$$
Thus, writing the sign vector  as in Figure \ref{Fig:SV}, we have the following  rule:
\begin{enumerate}
\item The  negative sign  flows as the arrows of Figure \ref{Fig:SV}. 
\item 
The forms $\varpi_{30}\cdot \phi$ and $\varpi_{29}\cdot \phi$ cannot both be positive at the same time.
\end{enumerate}
For example, if  $\varpi_{19}\cdot \phi$ is negative, the same is true of all the $\varpi_{i}\cdot \phi$ with $i=\{20,23,26,29,32,30\}$. The second rule is justified by the fact that  $\varpi_{30}+\varpi_{29}=-\alpha_6$ and we are restricted to the  dual Weyl chamber.
To define a chamber, we just need to name which one of the $\varpi_i\cdot \phi$ is the first negative one with respect to the order given above. If both $\varpi_{26}\cdot \phi$ and  $\varpi_{30}\cdot \phi$  are positive, 
then  $\varpi_{19}\cdot \phi$,  $\varpi_{20}\cdot \phi$, and  $\varpi_{23}\cdot \phi$ are all positive, 
 $\varpi_{29}\cdot \phi$ is necessarily negative (since  $\varpi_{30}\cdot \phi$ is positive), which forces $\varpi_{32}\cdot \phi$ to also be negative. 
There are exactly  eight possibilities satisfying these two rules.  They are listed in Table \ref{Table:Ch} and we check explicitly by solving inequalities that they all occur.  
Their adjacency graph is a Dynkin diagram of type E$_8$  illustrated in Figure \ref{fig:IE756} where we label the chambers by the corresponding simple roots of E$_8$.
\end{proof}

\section{Minimal models $Y_4$, $Y_5$, $Y_6$, and $Y_8$}

In this section, we  construct four distinct relative minimal models  over the Weierstrass model of an E$_7$-model as crepant resolutions of the Weierstrass model in equation \eqref{eq:E7}. We also study the flops between these distinct crepant resolutions. These relative minimal models correspond to  chambers Ch$_4$, Ch$_5$, Ch$_6$, and Ch$_8$ of the hyperplane arrangement I(E$_7$, $\bf{56}$)  discussed in section \ref{sec:IE756}. 
Accordingly, we denote these mimimal models by Y$_4$, Y$_5$, Y$_6$, and Y$_8$. We can tell them apart by identifying the extreme rays of de crepant resolution over the Weierstrass model. 
To each extreme ray $C$, we associate a unique weight of the representation $\bf{56}$ of E$_7$. 
The identification is given by computing   minus the intersection of the curve  $C$  with the fibral divisors $D_i$ ($i=1,\ldots, 7$) not touching the section of the elliptic fibration.
 The flops between the four minimal models Y$_4$, Y$_5$, Y$_6$, and Y$_8$ form a  Dynkin diagram of type D$_4$ as illustrated in Figure \ref{fig:IE756}.

\subsection{Overview of the sequence of blowups defining the resolutions}
Each of the models Y$_4$, Y$_5$, Y$_6$, and Y$_8$ can be obtained by numerous sequences of blowups. 
We give two distinct resolutions for Y$_4$ and Y$_5$ to ease the description of the flops. 
We will give two trees of blowups defining respectively crepant resolutions of the triples (Y$_4$, Y$_5$, Y$_8$) and (Y$_4$, Y$_6$,Y$_8$). 
The flops within each of these triples will be modeled by the flops of a suspended pinch point. 
The graph of flops of a suspended pinch point is a Dynkin diagram of type A$_3$ and gluing two such Dynkin diagram along two consecutive nodes gives a D$_4$-Dynkin diagram. For us, the two consecutive nodes will correspond to Y$_4$ and Y$_5$ and this explains why we give two distinct resolutions for these two minimal models.

\noindent{\bf Conventions for blowups.}
 Each crepant resolution is an embedded resolution defined by a sequence of blowups with smooth centers. 
We denote the blowup $X_{i+1}\to X_i$ along the ideal $(f_1,f_2,\ldots,f_n)$ with exceptional divisor $E$ as:
$$\begin{tikzcd}[column sep=2.4cm]X_i \arrow[leftarrow]{r} {\displaystyle (f_1,\ldots, f_n|E)}  & X_{i+1}\end{tikzcd},$$
where $X_0$ is the projective bundle in which the Weierstrass model is defined. 

 We abuse notations and call the proper transforms of the variables involved in the blowup by the same name. 
That means for example that we implement $(u_1, \ldots, u_n| e_1)$ by the birational transformation $(u_1, \ldots, u_n)\mapsto (u_1 e_1, \dots, u_n e_1)$ with the understanding that we introduce at the same time projective coordinates $[u_1:\ldots:u_n]$ which are the projective coordinates of the fiber of a $\mathbb{P}^{n-1}$-bundle introduced by the blowup.  In particular, after the blowup, $(u_1,\ldots, u_n)$ cannot vanish simultaneously as they are projective coordinates of a $\mathbb{P}^{n-1}$. 
If we blowup a regular sequence $(u_1,\ldots, u_n)$ and the defining equation has multiplicity $n-1$ along $V(u_1, \ldots, u_n)$, then by adjunction, it is easy to see that the embedded blowup defines a crepant map.

The sequence of blowups that we consider for the resolutions  Y$_4$ and  Y$_5$,  and to understand the flops between Y$_4$, Y$_5$, and Y$_8$, are the following
\begin{equation}\label{Ch458}
\begin{tikzpicture}[baseline= (a).base]
\node[scale=.9] (a) at (0,0) {
\begin{tikzcd}[column sep=1.4cm, ampersand replacement=\&]
\& \& \& \& \& \& \&  \text{X}_{7}^+ 
\\
X_0 \arrow[leftarrow]{r} {\displaystyle (x,y,s|e_1)} \& X_1 \arrow[leftarrow]{r} {\displaystyle (x,y,e_1|e_2)} \&  X_2\arrow[leftarrow]{r} {\displaystyle (y,e_1,e_2|e_3)} \&  X_3  \arrow[leftarrow]{r} {\displaystyle (y,e_2|e_4)} \&  
X_4 \arrow[leftarrow]{r} {\displaystyle (e_2,e_4|e_5)} \&  X_5  \arrow[leftarrow]{r} {\displaystyle (y,e_3|e_6)}  \arrow[leftarrow]{ddr}[left]{\displaystyle (e_3,e_4|e_6)}\&  X_6 
\arrow[leftarrow,sloped]{ru}[above]{\displaystyle (e_3,e_6|e_7)}  
\arrow[leftarrow,sloped]{rd}[below]{\displaystyle (e_3,e_4|e_7)}
\& \\
\& \& \& \& \& \& \& \text{X}_7^-\\
\& \& \& \& \& \&\text{X}_6' \arrow[leftarrow]{r}[above]{\displaystyle (y,e_3|e_7)} \& \text{X}_7'
\end{tikzcd} 
};
\end{tikzpicture}
\end{equation}
where $X_0$ is the projective bundle in which the Weierstrass model is defined. 
We work with embedded resolutions. 
The minimal models Y$_4$, Y$_5$, and Y$_8$ are  obtained as the proper transform of the Weierstrass model in $X_7^+$, $X_7^-$, and $X_7'$,     respectively.

To understand the flops between Y$_4$, Y$_5$, and Y$_6$, we consider the following tree of blowups:

\begin{equation}\label{Ch456}
\begin{tikzpicture}[baseline= (a).base]
\node[scale=.9] (a) at (0,0) {
\begin{tikzcd}[column sep=1.4cm, ampersand replacement=\&]
 \& \& \& \& \& \&  \text{X}_{6}^+ 
\\
X_0 \arrow[leftarrow]{r} {\displaystyle (x,y,s|e_1)} \& X_1
 \arrow[leftarrow]{r} {\displaystyle (y,e_1|e_2)} 
  \&  X_2
  \arrow[leftarrow]{r} {\displaystyle (x,y,e_2|e_3)} 
  \&  X_3 
   \arrow[leftarrow]{r} {\displaystyle (x,e_2,e_3|e_4)} \& X_4 \arrow[leftarrow]{r} {\displaystyle (e_2,e_4|e_5)}
   \arrow[leftarrow,sloped]{rdd}[below] {\displaystyle (e_2,e_3|e_5)}
    \&  X_5
\arrow[leftarrow,sloped]{ru}[above]{\displaystyle (e_2,e_5|e_6)}  
\arrow[leftarrow,sloped]{rd}[below]{\displaystyle (e_2,e_3|e_6)}
\& \\ 
\& \& \& \& \& \& \text{X}_6^-\\
 \& \& \& \& \&\text{X}_5' \arrow[leftarrow]{r}[above]{\displaystyle (e_2,e_4|e_6)} \& \text{X}_6'
\end{tikzcd} 
};
\end{tikzpicture}
\end{equation}
where 
Y$_4$, Y$_5$, and  Y$_6$ are obtained as the proper transform of the Weierstrass model in $X_6^+$ , $X_6^-$, and $X'_6$, respectively.

\noindent{\bf Conventions for Cartan divisors and fibers.}
We denote the fibral divisors by $D_m$ where $m=0,1, \ldots, 7$ and D$_0$ is the divisor touching the zero section of the elliptic fibration.  By definition, each $D_m$ is a fibration of a rational curve over the divisor $S$ in the base. We denote the generic curve of $D_m$ by  $C_m$. 
 The intersection numbers $\int_Y {D}_m\cdot {C}_n$ are minus the Cartan matrix of the affine Dynkin diagram of  type $\widetilde{\text{E}}_7$. The fibral divisors not touching the section are denoted $D_i$ with $i=1,2,\ldots, 7$. 
 We recall that the weight of a vertical curve $C$ is by definition the vector of minus its intersection number with the fibral divisors: 
 $$
 (-\int_Y D_0\cdot C, -\int_Y D_1\cdot C, \ldots, -\int_Y D_7\cdot C).
 $$
They are not independent, as the linear combination with coefficients  $(1,2,3,4,3,2,1,2)$ gives zero.\footnote{
These coefficients are the multiplicities of the node of the fiber III$^*$ and also correspond to the Dynkin labels of the highest root of E$_7$ with one given for the extra node.}
For that reason, it is enough to compute only the intersection with the D$_i$, which gives a vector in the weight lattice of E$_7$. 

\subsection{The geometry of  $Y_4$}
We study the minimal model $Y_4$ using the resolution discussed in equation \eqref{Ch458}. 
$Y_4$ is then the proper transform of the Weierstrass model of equation 
\eqref{eq:E7} after the blowups leading to X$_7^+$ in equation 
\eqref{Ch458}. The result is:
\begin{equation}
\text{Y}_4: \quad  e_4 e_6 y^2-e_1 e_2 e_3 (e_2 e_4 e_5^2  x^3+a e_1 s^3 x+b e_1^2 e_3 e_6 e_7^2 s^5)=0 .
\end{equation}
The relative projective coordinates are:
\begin{align}
\begin{aligned}
[e_2 e_3 e_4 e_5^2 e_6 e_7^2 x: e_2 e_3^2 e_4^2 e_5^3 e_6^3 e_7^5 y: s][x: 
  e_3 e_4 e_5 e_6^2 e_7^3 y: e_1 e_3 e_6 e_7^2]\\
   [e_4 e_5 e_6 e_7 y: e_1: 
  e_2 e_4 e_5^2][e_6 e_7 y: e_2 e_5][e_2: e_4][y: e_3 e_7] [e_3: e_6] .
\end{aligned}
\end{align}
The divisor for the special fiber is $se_1 e_2 e_3^2 e_4 e_5^2 e_6^2 e_7^4 $. The fibral divisors are:
\begin{equation}
\begin{cases}
1\   D_0: &\quad s= e_6 y^2- e_1 e_2^2 e_3 e_5^2 x^3 =0\\
2\  D_{1}: &\quad  e_1=y=0\\
3\  D_{2}: &\quad e_1=e_6=0\\
4\  D_{3}: &\quad e_7=e_4 e_6 y^2 -e_1 e_2^2 e_3 e_4 e_5^2 x^3 -a e_1^2 e_2 e_3 s^3 x =0\\
3\  D_4: &\quad  e_3=e_4=0\\
2\  D_{5}: & \quad e_1=e_4=0\\
    1\    D_6 : &\quad e_4= b e_1 e_3 e_6 e_7^2 s^2 + a x=0\\
       2\   D_7 :&\quad  
e_6=  e_2 e_4 e_5^2 x^2+a e_1 s^3=0 
\end{cases}
\end{equation}
We denote by $C_a$ the generic fiber of the fibral divisor $D_a$. 
All the fibral divisors are $\mathbb{P}^1$-bundles with the exception of $D_3$ and $D_6$. 

\begin{rem}
One might think that the curve $C_7$ degenerates at $a=0$. However, that is not the case because when $a=0$, the defining equation for $C_7$ gives $e_6=e_2 e_4 e_5^2 x^2=0$. Since $e_6$ and $e_2 e_4 x $ cannot vanish simultaneously, as is clear by looking at the projective coordinates, we deduce that over $V(a)\cap S$, the curve $C_7$ simplifies to  $C_{7}: a=e_6=e_4=0$ but does not degenerate. It follows that $D_0$, $D_1$, $D_2$, D$_4$, and $D_5$ are $\mathbb{P}^1$-projective bundles. 
\end{rem}
The degeneration at $a=0$ gives:  
\begin{align}
\begin{cases}
C_3\longrightarrow C_{36}+C'_{3}, \\
C_6 \longrightarrow 2C_{36}+C_{7}+C_4,\\
\end{cases}
\end{align}
where 
\begin{equation}
C_3':\quad a= e_7= e_6 y^2 -e_1 e_2^2 e_3  e_5^2 x^3=0 
 , \quad C_{36}:\  a=e_4=e_7=0.
\end{equation}
Only $C_{36}$ and $C_3'$ are new rational curves produced by the degeneration at $V(s,a)$. 
They are extreme rays. 
 The geometric weights of these curves are: 
 \begin{align}
\begin{cases}
C_{36} =\frac{1}{2}(C_6-C_4-C_7)  \to  \boxed{\ 0 \ 0\ 0  \ $-1$  \ 1 \ 0 \ $-1$ \ 1 } = \varpi_{26}, \\
C_{3}'  =\frac{1}{2}(2C_3-C_6+C_4+C_7)  \to \boxed{\ 0 \ 0\ 1  \ $-1$  \ 0 \ 0 \ 1 \ 0 } = -\varpi_{23}.
\end{cases}
\end{align}
The two weights $\varpi_{26}$ and $\varpi_{23}$ are weights of the representation $\bf{56}$ and uniquely characterize the chamber Ch$_4$ of the hyperplane arrangement I(E$_7$,$\mathbf{56}$).

Taking into account the multiplicities of the curves $C_4$, $C_6$, and $C_7$, we have 
\begin{equation}
4C_3 +C_6 \longrightarrow 6 C_{36} +4 C_3'+C_4+C_{7}.
\end{equation}
Together with the remaining curves,  we get   a fiber of type II$^*$ (with dual fiber the affine Dynkin diagram $\widetilde{\text{E}}_8$) but with the node  of multiplicity 5 (the node $\alpha_4$ of  $\widetilde{\text{E}}_8$) contracted to a point as illustrated on Figure \ref{Fig:Ch4Fib}.

\begin{rem}\label{Rem:Jum}
We also note that the fiber $C_6$ over $V(s,a,b)$ jumps in dimension and becomes a quadric surface. Thus the resolution of the E$_7$-fiber does not give a flat elliptic fibration when the base is a threefold and $a$ and $b$ can vanish simultaneously on $S$. No other fiber component jumps in dimension. 
\end{rem}

\begin{figure}[htb]
\begin{center}
\scalebox{.8}{
\begin{tikzpicture}
				\node[draw,circle,thick,scale=1,fill=black,label=below:{\scalebox{1.2}{ $\alpha_0$}}] (0) at (0,0){$1$};
				\node[draw,circle,thick,scale=1,label=below:{\scalebox{1.2}{$\alpha_1$}}] (1) at (1.2,0){$2$};
				\node[draw,circle,thick,scale=1,label=below:{\scalebox{1.2}{$\alpha_2$}}] (2) at (2.4,0){$3$};
				\node[draw,circle,thick,scale=1,label=below:{\scalebox{1.2}{$\alpha_3$}}] (3) at (3.6,0){$4$};
				\node[draw,circle,thick,scale=1,label=below:{\scalebox{1.2}{$\alpha_5$}}] (5) at (6,0){$6$};
				\node[draw,circle,thick,scale=1,label=below:{\scalebox{1.2}{$\alpha_6$}}] (6) at (7.2,0){$4$};
				\node[draw,circle,thick,scale=1, label=below:{\scalebox{1.2}{$\alpha_7$}}] (7) at (8.4,0){$2$};
				\node[draw,circle,thick,scale=1, label=above:{\scalebox{1.2}{$\alpha_8$}}] (8) at (6,1.6){$3$};
				\draw[thick] (0)--(1)--(2)--(3)--(5)--(6)--(7);
				\draw[thick]  (5)--(8);
					\end{tikzpicture}	
			}
\end{center}
\caption{The resolution in Chamber 4 has a fiber \~{E}$_7$ that degenerates to a fiber of type 
II$^*$ with the node $\alpha_4$ contracted to a point.
\label{Fig:Ch4Fib}}
\end{figure}

\subsection{The geometry of $Y_5$}

We study the minimal model Y$_5$ using the resolution discussed in equation \eqref{Ch458}. 
Y$_5$ is the proper transform of the Weierstrass model of equation 
\eqref{eq:E7} after the blowups leading to X$_7^-$ in equation 
\eqref{Ch458}. The result is:
\begin{equation}
Y_{5}: \quad  e_4 e_6 y^2= e_1 e_2 e_3 (e_2 e_4 e_5^2 e_7 x^3+a e_1 s^3 x+b e_1^2 e_3 e_6 e_7 s^5 ).
\end{equation}
The relative projective coordinates due to the successive blowups are 
\begin{equation}
\begin{aligned}
[e_2 e_3 e_4 e_5^2 e_6 e_7^2 x: e_2 e_3^2 e_4^2 e_5^3 e_6^3 e_7^4 y: s] 
[x: e_3 e_4 e_5 e_6^2 e_7^2 y: e_1 e_3 e_6 e_7] \\
[e_4 e_5 e_6 e_7 y: e_1: e_2 e_4 e_5^2 e_7]
[e_6 y: e_2 e_5]
[e_2: e_4 e_7 ]
[y: e_3 e_7] [e_3:e_4].
\end{aligned}
\end{equation}
The divisor for the special fiber is $s e_1 e_2 e_3^2 e_4 e_5^2 e_6^2 e_7^3$ with irreducible components:
\begin{equation}
\begin{aligned}
\begin{cases}
1\ D_0:\quad  & s=e_1 e_2^2 e_3 e_5^2 e_7 x^3 - e_6 y^2=0 \\
2\ D_1: \quad &e_1=y=0\\ 
3\ D_2: \quad &e_1=e_6=0 \\ 
4\ D_3: \quad &e_3=e_6=0 \\ 
3\ D_4: \quad &   e_7=e_4 e_6 y^2-a e_1^2 e_2 e_3 s^3 x=0 \\ 
2\ D_5: \quad  &e_5=e_4 e_6 y^2- a e_1^2 e_2 e_3 s^3 x-b e_1^3 e_2 e_3^2 e_6 e_7 s^5 =0\\
1\ D_6: \quad & e_4=a x+b e_1 e_3 e_6 e_7 s^2 =0  \\ 
2\ D_7: \quad & e_6=e_2 e_4 e_5^2 e_7 x^2+a e_1 s^3 =0
\end{cases}
\end{aligned}
\end{equation}
All the fibral divisors are $\mathbb{P}^1$-bundles except for $D_4$, $D_6$ and $D_7$ which  degenerate over $V(s,a)$:
\begin{align}
\begin{cases}
C_4\longrightarrow C_{46}+C_{47}, \\
C_6 \longrightarrow C_{46}+C_{67},\\
C_{7}\longrightarrow C_{47}+C_{67}
\end{cases}
\end{align}
where:
\begin{equation}
C_{46}:\quad a= e_4= e_7=0, \quad C_{47}:\quad a= e_6= e_7=0, \quad C_{67}:\quad a= e_4= e_6=0.
\end{equation}
The naming is chosen to reflect the intersection of the original curves rather than our choice of blowup (over a point of $V(s,a),$  $C_{ij}$   is the intersection of $D_i$ and $D_j$). Taking into account the multiplicities of the curves $C_4$, $C_6$, and $C_7$, we have 
\begin{equation}
3 C_4+ 2C_7+ C_6\longrightarrow 5C_{47}+4 C_{46}+3C_{67}.
\end{equation}
The three new curves ($ C_{46}, C_{47}$, and $C_{67}$) all intersect at the same point $a=e_4=e_6=e_7=0$. 
The curve $C_{67}$ does not intersect any other curves since $C_{67}$ is only visible in the patch 
$e_5e_3e_2 x s \neq 0$. 
The curve $C_{47}$ intersects transversally the curve $C_3$ and the curve $C_{46}$ intersects transversally the curve $C_5$. 
Altogether, the curves form  a fiber of type II$^*$ with the nodes  of multiplicity 6 ($\alpha_5$) contracted to  a point as illustrated in Figure \ref{Fig:Ch5Fib}.

\begin{figure}[htb]
\begin{center}
\scalebox{.8}{
\begin{tikzpicture}
				\node[draw,circle,thick,scale=1,fill=black,label=below:{\scalebox{1.2}{ $\alpha_0$}}] (0) at (0,0){$1$};
				\node[draw,circle,thick,scale=1,label=below:{\scalebox{1.2}{$\alpha_1$}}] (1) at (1.2,0){$2$};
				\node[draw,circle,thick,scale=1,label=below:{\scalebox{1.2}{$\alpha_2$}}] (2) at (2.4,0){$3$};
				\node[draw,circle,thick,scale=1,label=below:{\scalebox{1.2}{$\alpha_3$}}] (3) at (3.6,0){$4$};
				\node[draw,circle,thick,scale=1,label=below:{\scalebox{1.2}{$\alpha_4$}}] (4) at (4.8,0){$5$};
				\node[draw,circle,thick,scale=1,label=below:{\scalebox{1.2}{$\alpha_6$}}] (6) at (7.2,0){$4$};
				\node[draw,circle,thick,scale=1, label=below:{\scalebox{1.2}{$\alpha_7$}}] (7) at (8.4,0){$2$};
				\node[draw,circle,thick,scale=1, label=above:{\scalebox{1.2}{$\alpha_8$}}] (8) at (6,1.6){$3$};
				\draw[thick] (0)--(1)--(2)--(3)--(4)--(6)--(7);
				\draw[thick]  (6,0)--(8);
					\end{tikzpicture}	
					}
					\end{center}
\caption{
The resolution Y$_5$ (corresponding to Chamber 5) has a fiber $\widetilde{\text{E}}_7$ that degenerates to a fiber of type 
II$^*$ (with dual graph $\widetilde{\text{E}}_8$) with the node $\alpha_4$ contracted to a point.
\label{Fig:Ch5Fib}
}
\end{figure}
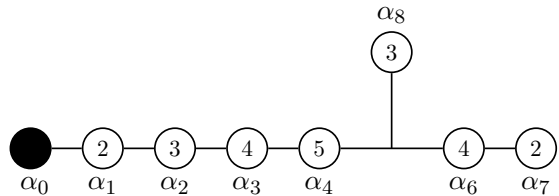   
The corresponding weights are 
\begin{align}
\begin{cases}
& C_{46}=\frac{1}{2}(C_4+C_6-C_7)\to \boxed{\ 0 \ 0\ 0 \ 0 \ $-1$  \ 1 \ $-1$ \ 1 } = \varpi_{29}, \\
& C_{47}=\frac{1}{2}(C_4+C_7-C_6)\to  \boxed{\ 0 \ 0\ 0 \ 1 \ $-1$  \ 0 \ 1 \ $-1$ } = -\varpi_{26},\\
& C_{67}=\frac{1}{2}(C_6+C_7-C_4)\to  \boxed{\ 0 \ 0\ 0 \ 0 \ 1  \ 0 \ $-1$ \ $-1$ } = \varpi_{30}.
\end{cases}
\end{align}
After removing the first component (corresponding to $C_0$), these vectors become weights in the representation $\bf{56}$ of E$_7$. The Weyl orbit of each of these weights produces the full representation $\mathbf{56}$ of E$_7$ since it is a minuscule representation. 
The hyperplanes $\varpi_{26}^\bot$, $\varpi_{29}^\bot$, $\varpi_{30}^\bot$ are the interior walls of the chamber Ch$_5$ (see Figure \ref{fig:IE756}).
As in Remark \ref{Rem:Jum}, the node C$_6$ becomes a rational surface over $V(s,a,b)$.

\subsection{The geometry of  $Y_6$}\label{sec:Y6}

The minimal model Y$_6$ is discussed using the resolution in \eqref{Ch456}. 
The proper transform  is
\begin{equation}\label{Y6Res}
Y_6:\quad e_2 e_3  e_5y^2-e_1  (e_3 e_6 e_4^2   x^3+a e_1 e_2  s^3 x+b e_1^2 e_2^2 e_5  e_6s^5)=0.
\end{equation}
The projective coordinates are
\begin{equation}\label{Y6projcoord}
\begin{aligned}
  [e_3 e_4^2 e_5e_6^2  x:e_2 e_3^2 e_4^3 e_5^3 e_6^4 y:s] 
    [e_3 e_4 e_5  e_6 y:e_1]  [e_4  e_6 x:y:e_2 e_4 e_5 e_6^2]
      [ x:e_2 e_5e_6 :e_3 e_5] \\
      [e_2e_6 :e_3][e_2:e_4] .  
\end{aligned}
\end{equation}
The total transform of $s$ is 
$s e_1 e_2 e_3 e_4^2 e_5^2 e_6^3 $ and the fibral divisors  are:

\begin{equation}
\begin{cases}
1\   D_0: &\quad s= e_2   e_5y^2-e_1  e_6 e_4^2   x^3=0\\
2\  D_{1}: &\quad e_1=e_2=0\\
3\  D_2: &\quad  e_2=x=0\\
4\  D_{3}: &\quad e_2=e_6=0
\\
3\  D_{4}: &\quad  e_6=  e_3  e_5y^2-a e_1^2   s^3 x=0 \\
2\  D_{5}: & \quad  e_4=e_3  e_5y^2-e_1^2 s^3  (a x+b e_1 e_2 e_5  e_6s^2)=0\\
    1\    D_6 : &\quad e_3=a x+b e_1 e_2 e_5e_6 s^2=0 \\
           2\   D_7 :&\quad  e_5=e_3 e_4^2e_6   x^2+a e_1 e_2s^3=0
\end{cases}
\end{equation}
The resolution Y$_6$ (corresponding to Chamber 6) has a fiber III$^*$ (with dual graph $\widetilde{\text{E}}_7$) over the generic point of $S$. But  over $V(s,a)$, the fiber III$^*$ degenerates as follows:
\begin{align}
\begin{cases}
 C_5\longrightarrow C_{57}+C_5',\\
  C_7\longrightarrow C_{4}+ C_{6}+2C_{57} .
\end{cases}
\end{align}
We have 
\begin{align}
2C_5+2C_7\longrightarrow 2C_{4}+6 C_{57}+2C_5'+2 C_{6},
\end{align}
where
\begin{equation}
\begin{aligned}
&C_{57}:a=e_4=e_5=0, 
\quad C_5': e_4=e_3  y^2-b e_1^3 e_2 e_6s^5=0.
\end{aligned}
\end{equation}
The resulting fiber is a fiber of type  II$^*$ (with dual graph $\widetilde{\text{E}}_8$) with the node $\alpha_6$ contracted to a point as illustrated in Figure \ref{Fig:Ch6Fib}.
The weights of these curves can be computed as follows
\begin{align}
\begin{cases}
C_{57}=\frac{1}{2} (C_7-C_4-C_6) \rightarrow   \boxed{0\ 0\   0\  1\  $-1$ \  1\  $-1$ \  0} = -\varpi_{29},\\
C_5' = \frac{1}{2}(2C_5-C_7+C_4+C_6) \rightarrow  \boxed{0\ 0\   0\  0\  0\  $-1$ \  0\  1} = \varpi_{32}.
 \end{cases}
 \end{align}
 These are weights of the representation $\bf{56}$ of E$_7$. 
 The hyperplanes $\varpi_{29}^\bot$ and $\varpi_{32}^\bot$ are walls of the chamber $\mathrm{Ch}_6$ of the hyperplane arrangement I(E$_7$, $\bf{56}$). 

As for Y$_4$ and Y$_5$, when the base has at least dimension three, the curve $C_6$ becomes a quadric surface over $V(s,a,b)$. 

\begin{figure}[htb]
\begin{center}
	\scalebox{.8}{
		\begin{tikzpicture}
				% E8
				\node[draw,circle,thick,scale=1,fill=black,label=below:{\scalebox{1.2}{ $\alpha_0$}}] (0) at (0,0){$1$};
				\node[draw,circle,thick,scale=1,label=below:{\scalebox{1.2}{$\alpha_1$}}] (1) at (1.2,0){$2$};
				\node[draw,circle,thick,scale=1,label=below:{\scalebox{1.2}{$\alpha_2$}}] (2) at (2.4,0){$3$};
				\node[draw,circle,thick,scale=1,label=below:{\scalebox{1.2}{$\alpha_3$}}] (3) at (3.6,0){$4$};
				\node[draw,circle,thick,scale=1,label=below:{\scalebox{1.2}{$\alpha_4$}}] (4) at (4.8,0){$5$};
				\node[draw,circle,thick,scale=1,label=below:{\scalebox{1.2}{$\alpha_5$}}] (5) at (6,0){$6$};
							\node[draw,circle,thick,scale=1, label=below:{\scalebox{1.2}{$\alpha_7$}}] (7) at (8.4,0){$2$};
				\node[draw,circle,thick,scale=1, label=above:{\scalebox{1.2}{$\alpha_8$}}] (8) at (6,1.6){$3$};
				\draw[thick] (0)--(1)--(2)--(3)--(4)--(5)--(7);
				\draw[thick]  (5)--(8);
					\end{tikzpicture}}
\end{center}
\caption{The resolution in Chamber 6 has a fiber $\widetilde{\text{E}}_7$ that degenerates to a fiber of type 
II$^*$ with the node $\alpha_6$ contracted to a point.
\label{Fig:Ch6Fib}}
\end{figure}

\subsection{The geometry of  $Y_8$}
We study the minimal model Y$_8$ using the resolution discussed in equation \eqref{Ch458}. 
Y$_5$ is the proper transform of the Weierstrass model of equation 
\eqref{eq:E7} after the blowups leading to X$_7^-$ in equation 
\eqref{Ch458}. The result is:
\begin{equation}\label{eq:Y8}
Y_8:\quad\quad
e_4 e_7 y^2-e_1 e_2 e_3 (e_2 e_4 e_5^2 e_6 x^3+a e_1 s^3 x+b e_1^2 e_3 e_6 e_7 s^5)=0.
\end{equation}
where the relative projective coordinates are
\begin{align}\label{eq:Y8projcoord}
\begin{aligned}
[e_2 e_3 e_4 e_5^2 e_6^2 e_7 x : e_2 e_3^2 e_4^2 e_5^3 e_6^4 e_7^3 y : s]\   [x :
  e_3 e_4 e_5 e_6^2 e_7^2 y : e_1 e_3 e_6 e_7]\\
   [e_4 e_5 e_6 e_7 y : e_1 : e_2 e_4 e_5^2 e_6]\  
   [e_7 y : e_2 e_5]\   [e_2 : e_4 e_6]\  [e_4 : e_3 e_7] \  [y : e_3].
  \end{aligned}
\end{align}
The total transform of $s$ is 
$se_1 e_2 e_3^2 e_4 e_5^2 e_6^3 e_7^2 $ and we have the following fibral divisors 
\begin{align}
\begin{cases}
1\ D_0:\quad  & s=e_1 e_2^2 e_3 e_5^2 e_6 x^3 - e_7 y^2=0 \\
2\ D_1: \quad &e_1=y=0\\ 
3\ D_2: \quad &e_1=e_7=0 \\ 
4\ D_3: \quad &e_3=e_7=0 \\ 
3\ D_4: \quad &   e_6=e_4 e_7 y^2-a e_1^2 e_2 e_3 s^3 x=0 \\ 
2\ D_5: \quad  &e_5=e_4 e_7 y^2- a e_1^2 e_2 e_3 s^3 x-b e_1^3 e_2 e_3^2 e_6 e_7 s^5 =0\\
1\ D_6: \quad & e_4=a x+b e_1 e_3 e_6 e_7 s^2 =0  \\ 
2\ D_7: \quad & e_7=e_2 e_4 e_5^2 e_6 x^2+a e_1 s^3 =0
\end{cases}
\end{align}
At $V(s,a)$, we have:
\begin{align}
C_4  \longrightarrow   C_{6}+C_{7} +2 C'_{4},     
\end{align}
where
\begin{equation}
C'_4: a=e_6=y=0.
\end{equation}
Only C$'_4$ contributes a weight that is not in the root lattice: 
\begin{equation}
 C'_{4}\longrightarrow
 \frac{1}{2}(C_4-C_6-C_7)\longrightarrow   \boxed{0\  \ 0\   0\   0\  -1\   0\  1\ 1}= -\varpi_{30}. 
\end{equation}
The weight $\varpi_{30}$ is in the representation $\bf{56}$ of E$_7$. 
The hyperplane $\varpi_{30}^\bot$ perpendicular  to $\varpi_{30}$ is the interior wall of chamber Ch$_8$ of the hyperplane arrangement  I(E$_7$,$\bf{56}$). 
The fiber obtained at the degeneration $V(s,a)$ is a one-chain  with the following multiplicities:
\begin{equation}
C_0-2C_1-3C_2-4C_3-5C_{47}-6 C'_{4}-4C_{46}-2 C_5 .
\end{equation}
This chain corresponds to an affine Dynkin diagram of type \~E$_8$ with the node $\alpha_8$ removed as illustrated on Figure \ref{Fig:Ch8Fib}. 
\begin{figure}[htb]
\begin{center}
					\scalebox{.8}{
					\begin{tikzpicture}
				% E8
				\node[draw,circle,thick,scale=1,fill=black,label=below:{\scalebox{1.2}{ $\alpha_0$}}] (0) at (0,0){$1$};
				\node[draw,circle,thick,scale=1,label=below:{\scalebox{1.2}{$\alpha_1$}}] (1) at (1.2,0){$2$};
				\node[draw,circle,thick,scale=1,label=below:{\scalebox{1.2}{$\alpha_2$}}] (2) at (2.4,0){$3$};
				\node[draw,circle,thick,scale=1,label=below:{\scalebox{1.2}{$\alpha_3$}}] (3) at (3.6,0){$4$};
				\node[draw,circle,thick,scale=1,label=below:{\scalebox{1.2}{$\alpha_4$}}] (4) at (4.8,0){$5$};
				\node[draw,circle,thick,scale=1,label=below:{\scalebox{1.2}{$\alpha_5$}}] (5) at (6,0){$6$};
				\node[draw,circle,thick,scale=1,label=below:{\scalebox{1.2}{$\alpha_6$}}] (6) at (7.2,0){$4$};
				\node[draw,circle,thick,scale=1, label=below:{\scalebox{1.2}{$\alpha_7$}}] (7) at (8.4,0){$2$};
				%\node[draw,circle,thick,scale=1, label=above:{\scalebox{1.2}{$\alpha_8$}}] (8) at (6,1.6){$3$};
				\draw[thick] (0)--(1)--(2)--(3)--(4)--(5)--(6)--(7);
				%\draw[thick]  (5)--(8);
					\end{tikzpicture}}			
\end{center}
\caption{The resolution in Chamber 8 has a fiber $\widetilde{\text{E}}_7$ that degenerates to a fiber of type 
II$^*$ with the node $\alpha_8$ contracted to a point.
\label{Fig:Ch8Fib}}
\end{figure}

\subsection{SPP (suspended pinch point) flops\label{sec:SSP}}

  Binomial hypersurfaces are simple algebraic varieties. They can be  instrumental in our understanding of flops between  minimal models over a Weierstrass model. 
For example, the binomial variety in $\mathbb{C}^5$ defined by the hypersurface
$$
u_1 u_2 -w_1 w_2 w_3=0,
$$  
has six crepant resolutions whose graph of flops  is a hexagon (an affine Dynkin diagram of type $\widetilde{\text{A}}_5$). 
The flop diagram of this binomial variety matches those of the Spin($8$)-model \cite{G2} and also defines the hexagon of flops of the SU($5$) model 
 \cite{EY}. 
 
The {\em pinch point }(also called the {\em Whitney umbrella}) is the singular surface defined by the following binomial equation in $\mathbb{C}^3$:
$$
 x^2-y^2 z=0.
$$
The Whitney umbrella is not a normal surface as it has singularities in codimension-one. The Whitney umbrella plays an important role in the geometry of weak coupling limits \cite{CDE,AE1,AE2,Esole:2012tf}.

The suspended pinch  point appears in other areas of string geometry \cite{Morrison:1998cs}.   The {\em suspended pinch point} is a threefold defined as the double cover of $\mathbb{C}^3$ branched along a Whitney umbrella. Its defining equation in $\mathbb{C}^4$ is $$u^2=x^2-y^2 z.$$ 
After a change of variables $(u, x, y, z)\to (u_1=z, u_2=u+x, u_3=u-x, u_4=y)$, the suspended pinch point becomes  the following binomial hypersurface in $\mathbb{C}^4$:
$$
Z_0: u_1 u_4^2 = u_2 u_3. 
$$
The singularities of the suspended pinch point are in codimension-two, more precisely, the singular locus is the line $u_2=u_3=u_4=0$ and the singularity worsens at the origin. The variety is  normal and has three distinct crepant resolutions whose graph of flops is a Dynkin diagram of type A$_3$. 
The three crepant resolutions of the suspended pinch points are presented algebraically in Figure \ref{Fig:SSPRes}. 
We also give a toric description of all the crepant partial resolutions in Figure \ref{Figure:SSPToric}.
The algebraic crepant resolutions given in Figure  \ref{Fig:SSPRes} should be compared to the right tails of the trees of resolutions presented in 
equations \eqref{Ch458} and \eqref{Ch456}. This explains why the flops between the minimal models $Y_4$, $Y_5$, $Y_6$, and $Y_8$ are inspired by those of the suspended pinch point.

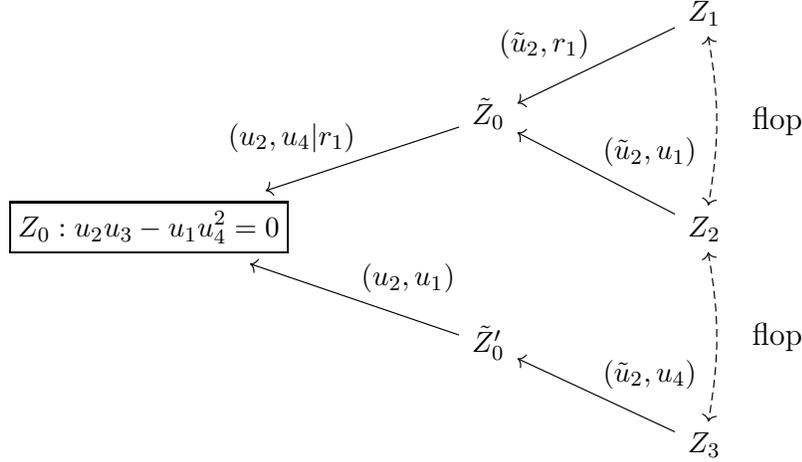
\begin{figure}[htb]
\begin{center}
\begin{tikzcd}[column sep=60pt]
    &  & Z_1  
         \arrow[bend left=10, leftrightarrow, dashed]{dd}[right] {\quad \text{\large flop}}
    &\\ 
    & {
    \tilde{Z}_0}\arrow[leftarrow]{ru} {\displaystyle (\tilde{u}_2,r_1)}  \arrow[leftarrow]{rd} {\displaystyle (\tilde{u}_2,u_1)}& &\\
  \boxed{{Z}_0:u_2 u_3-u_1 u_4^2=0}\arrow[leftarrow]{ru} {\displaystyle (u_2,u_4|r_1)}   \arrow[leftarrow]{rd} {\displaystyle (u_2,u_1)}&  & Z_2  \arrow[bend left=10, leftrightarrow, dashed]{dd}[right] {\quad \text{\large flop}}&\\
     &
   {\tilde{Z}'_0}   \arrow[leftarrow]{rd} {\displaystyle (\tilde{u}_2,u_4)}&  \\
   & & Z_3
\end{tikzcd} 
\end{center}
\caption{The three crepant resolutions of the suspended pinch point $Z_0:u_2 u_3-u_1 u_4^2=0$. 
In the blowup defining $Z_1$,  $r_1=0$ is the exceptional locus of the previous blowup defining $\tilde{Z}_0$, i.e. $r_1$ is the exceptional locus of the blowup of $Z_0$  centered at $u_2=u_4=0$. 
The flops between $Z_1$, $Z_2$, and $Z_3$ defines a Dynkin diagram of type A$_3$ and are all Atiyah flops as shown in Figure \ref{Figure:SSPToric}.
\label{Fig:SSPRes}}
\end{figure}

Blowing up $(u_1, u_2)$ by  $(u_1, u_2)\to (r_1 \tilde{u}_1, r_1 \tilde{u}_2)$ with exceptional divisor $r_1=0$, the proper transform of $Z_0$ is 
$$
\tilde{Z}_0: \tilde{u}_2 u_3-\tilde{u}_1 u_4^2=0.
$$
The variables  $[\tilde{u}_1 :\tilde{u}_2]$ are the projective coordinates of a $\mathbb{P}^1$,  in particular, $\tilde{u}_1$ and $\tilde{u}_2$ cannot vanish at the same time. 
It follows that $\tilde{Z}_0$  is singular  at $u_4=\tilde{u}_2= u_3=0$, which is defined in the patch $\tilde{u}_1\neq 0$. In that patch,  $\tilde{Z}_0$ is isomorphic to the cylindrical quadric cone  $$(\mathbb{C}^2/\mathbb{Z}_2)\times\mathbb{C},$$
where the  $\mathbb{Z}_2$ involution is generated by minus the identity of $\mathbb{C}^2$. As a hypersurface in $\mathbb{C}^4$, it is given by   the zero locus of 
$$x^2- y z=0,$$
where   $\mathbb{C}^4$ is parametrized by $(r_1, x=u_4, y=\frac{\tilde{u}_2}{\tilde{u}_1},z=u_3)$. 
This cone   has a unique crepant resolution obtained, for example, by blowing up $(x,y)$. For $\tilde{Z}_0$, we blowup  $(u_4,u_2)$ using the birational map  $(u_4,u_2)\to (\tilde{u}_4 r_2,\tilde{u}_2 r_2)$,   
where $r_2=0$ is the exceptional divisor. The proper transform is  the crepant resolution
$$
Z_3:    r_2 \tilde{u}_1 \tilde{u}_4^2-\tilde{u}'_2 u_3=0.
$$
We get the other two crepant  resolutions of $Z_0$ by first blowing up $(u_2,u_4)$ with the birational map 
$(u_4,u_2)\to (\tilde{u}_4 r_1,\tilde{u}_2 r_1)$, which gives the proper transform  
$$
\tilde{Z}'_0:   {u}_1 r_1 \tilde{u}_4^2-\tilde{u}_2 {u}_3=0.
$$ 
The variables  $[\tilde{u}_4:\tilde{u}_2]$ are the projective coordinates of a  $\mathbb{P}^1$, in particular,  $\tilde{u}_4$ and $\tilde{u}_2$ cannot vanish at the same time. 
The singular locus of $Z'_0$ is $u_1= r_1=u_2=u_3=0$. This singularity is defined in the patch $\tilde{u}_4\neq 0$  and is a quadric cone    $$k[x,y,z,t]/(x y- z t),$$ 
where $x=u_1$, $y=r_1$, $z=\tilde{u}_2 /\tilde{u}_4$, and $\tilde{u}_3\tilde{u}_4$. 
Such a quadric cone has two  crepant resolutions connected by an Atiyah flop and obtained by blowing up  $(x,z)$ and $(x,t)$,  respectively. 
For $\tilde{Z}'_0$, we obtain the  two crepant resolutions $Z_1$ and $Z_2$ by blowing up $\tilde{Z}'_0$ at $(\tilde{u}_2,r_1)$  and $(\tilde{u}_2,u_1)$,   respectively.

The suspended pinch point  can also be described as a
 particular partial resolution of the $\mathbb{Z}_2\times \mathbb{Z}_2$ orbifold quotient of $\mathbb{C}^3$ \cite{Morrison:1998cs}:
$$
\mathbb{C}^3/(\mathbb{Z}_2\times\mathbb{Z}_2):\quad u_4^2- u_1 u_2 u_3=0.
$$
where the Klein group $\mathbb{Z}_2\times\mathbb{Z}_2$ is generated by the two involution $(z_1, z_2, z_3)\mapsto (-z_1, -z_2, z_3)$  and $(z_1, z_2, z_3)\mapsto (z_1,- z_2, -z_3)$.
Consider the  blowup centered at $(u_4,u_1)$.  In one patch, we have $(u_4,u_1)\to (u_4 u_1, u_1)$ which gives the proper transform 
$
u_1 u_4^2 = u_2 u_3.
$
 The diagram of flops of the crepant resolutions of this orbifold is also a D$_4$-Dynkin diagram as seen in Figure \ref{Fig:C322}.
 
 \clearpage  
 \begin{figure}[htb]
\begin{center}
			\scalebox{.8}{\begin{tikzpicture}
			\draw[dashed, latex-latex, line width=1 pt] (1.2,0)--(3.2,0);  \draw[dashed, latex-latex, line width=1 pt]  (6.2,0)--(8.2,0);
						\draw [-latex, line width=1 pt] (.2,-1)--+(-45:1.8);
												\draw[->,->=stealth, ,line width=1 pt] (2.5,-4.1)--+(-45:1.7);
												\draw [-latex, line width=1 pt](7,-4.1)--+(225:1.7);
												\draw [-latex, line width=1 pt](9.5,-1.1)--+(225:1.7);
												\draw [-latex, line width=1 pt](4,-1.1)--+(225:1.7);\draw[-latex, line width=1 pt] (6,-1.1)--+(-45:1.7);
												\draw [-latex, line width=1 pt](12,-4.1)--+(200:5);
												\draw [-latex, line width=1 pt](-2,-4.1)--+(-20:5);
												\draw [-latex, line width=1 pt](-.5,-1.1)--+(225:1.7);
												\draw [-latex, line width=1 pt](10.5,-1.1)--+(-45:1.7);

					\node (1) at (0,0)	{\begin{tikzpicture}
				\node[draw,circle,thick,scale=1,fill=black,label=above:{}] (1) at (0,0){};
				\node[draw,circle,thick,scale=1,fill=black,label=above:{}] (2) at (1,0){};
				\node[draw,circle,thick,scale=1,fill=black,label=above:{}] (3) at (2,0){};
				\node[draw,circle,thick,scale=1,fill=black,label=above:{}] (5) at (0,1){};
				\node[draw,circle,thick,scale=1,fill=black,label=above:{}] (4) at (1,1){};
				\draw[thick] (1) to (2) to (3) to (4) to (5)  to (1);
				\draw[thick] (1) to (4) to (2);
			\end{tikzpicture}};
					\node (1) at (5,0)	{\begin{tikzpicture}
				\node[draw,circle,thick,scale=1,fill=black,label=above:{}] (1) at (0,0){};
				\node[draw,circle,thick,scale=1,fill=black,label=above:{}] (2) at (1,0){};
				\node[draw,circle,thick,scale=1,fill=black,label=above:{}] (3) at (2,0){};
				\node[draw,circle,thick,scale=1,fill=black,label=above:{}] (5) at (0,1){};
				\node[draw,circle,thick,scale=1,fill=black,label=above:{}] (4) at (1,1){};
				\draw[thick] (1) to (2) to (3) to (4) to (5)  to (1);
				\draw[thick] (2) to (4);   \draw[thick] (5) to (2);
			\end{tikzpicture}};
			\node (3) at (10,0) 
			{
			\begin{tikzpicture}
				\node[draw,circle,thick,scale=1,fill=black,label=above:{}] (1) at (0,0){};
				\node[draw,circle,thick,scale=1,fill=black,label=above:{}] (2) at (1,0){};
				\node[draw,circle,thick,scale=1,fill=black,label=above:{}] (3) at (2,0){};
				\node[draw,circle,thick,scale=1,fill=black,label=above:{}] (5) at (0,1){};
				\node[draw,circle,thick,scale=1,fill=black,label=above:{}] (4) at (1,1){};
				\draw[thick] (1) to (2) to (3) to (4) to (5)  to (1);
				\draw[thick] (5) to (3);
				\draw[thick] (5) to (2);
			\end{tikzpicture}};

			\node (4) at (2,-3) {
			\begin{tikzpicture}
				\node[draw,circle,thick,scale=1,fill=black,label=above:{}] (1) at (0,0){};
				\node[draw,circle,thick,scale=1,fill=black,label=above:{}] (2) at (1,0){};
				\node[draw,circle,thick,scale=1,fill=black,label=above:{}] (3) at (2,0){};
				\node[draw,circle,thick,scale=1,fill=black,label=above:{}] (5) at (0,1){};
				\node[draw,circle,thick,scale=1,fill=black,label=above:{}] (4) at (1,1){};
				\draw[thick] (1) to (2) to (3) to (4) to (5)  to (1);
				\draw[thick] (2) to (4);
			\end{tikzpicture}};
			
			\node (5) at (8,-3) {
						\begin{tikzpicture}
				\node[draw,circle,thick,scale=1,fill=black,label=above:{}] (1) at (0,0){};
				\node[draw,circle,thick,scale=1,fill=black,label=above:{}] (2) at (1,0){};
				\node[draw,circle,thick,scale=1,fill=black,label=above:{}] (3) at (2,0){};
				\node[draw,circle,thick,scale=1,fill=black,label=above:{}] (5) at (0,1){};
				\node[draw,circle,thick,scale=1,fill=black,label=above:{}] (4) at (1,1){};
				\draw[thick] (1) to (2) to (3) to (4) to (5)  to (1);
				\draw[thick] (2) to (5);
			\end{tikzpicture}};
			
			\node (6) at (5,-6){
			\begin{tikzpicture}
				\node[draw,circle,thick,scale=1,fill=black,label=above:{}] (1) at (0,0){};
				\node[draw,circle,thick,scale=1,fill=black,label=above:{}] (2) at (1,0){};
				\node[draw,circle,thick,scale=1,fill=black,label=above:{}] (3) at (2,0){};
				\node[draw,circle,thick,scale=1,fill=black,label=above:{}] (5) at (0,1){};
				\node[draw,circle,thick,scale=1,fill=black,label=above:{}] (4) at (1,1){};
				\draw[thick] (1) to (2) to (3) to (4) to (5)  to (1);				
			\end{tikzpicture}
			};

			\node (7) at (12,-3) {
						\begin{tikzpicture}
				\node[draw,circle,thick,scale=1,fill=black,label=above:{}] (1) at (0,0){};
				\node[draw,circle,thick,scale=1,fill=black,label=above:{}] (2) at (1,0){};
				\node[draw,circle,thick,scale=1,fill=black,label=above:{}] (3) at (2,0){};
				\node[draw,circle,thick,scale=1,fill=black,label=above:{}] (5) at (0,1){};
				\node[draw,circle,thick,scale=1,fill=black,label=above:{}] (4) at (1,1){};
				\draw[thick] (1) to (2) to (3) to (4) to (5)  to (1);
				\draw[thick] (3) to (5);
			\end{tikzpicture}};

			\node (8) at (-2,-3) {
						\begin{tikzpicture}
				\node[draw,circle,thick,scale=1,fill=black,label=above:{}] (1) at (0,0){};
				\node[draw,circle,thick,scale=1,fill=black,label=above:{}] (2) at (1,0){};
				\node[draw,circle,thick,scale=1,fill=black,label=above:{}] (3) at (2,0){};
				\node[draw,circle,thick,scale=1,fill=black,label=above:{}] (5) at (0,1){};
				\node[draw,circle,thick,scale=1,fill=black,label=above:{}] (4) at (1,1){};
				\draw[thick] (1) to (2) to (3) to (4) to (5)  to (1);
				\draw[thick] (1) to (4);
			\end{tikzpicture}};

\end{tikzpicture}
}
\end{center}
\caption{Flops between the three crepant resolutions of the suspended pinch point $z u^2=   x y $ and all its toric partial resolutions.
The suspended pinch point is at the bottom row. The three crepant resolutions are on the  top row.  The four partial resolutions are in the middle row. 
The external varieties of the middle row have the singularities of a cylindrical quadric cone $k[x,y,z,t]/(z^2 - x y)$ while the others have singularities of a quadric threefold with a double point  $k[x,y,z,t]/(xy-zt)$. 
\label{Figure:SSPToric}
}

\begin{center}
\scalebox{.8}{
\begin{tikzpicture}
			\draw[dashed, latex-latex, line width=1 pt] (1.2,0)--(3.2,0);  \draw[dashed, latex-latex, line width=1 pt]  (6.2,0)--(8.2,0); \draw[dashed, latex-latex, line width=1 pt] (4.3,3.5)--(4.3,1.5);
					
					\node (1) at (0,0)	{\begin{tikzpicture}
				\node[draw,circle,thick,scale=1,fill=black,label=above:{}] (6) at (0,2){};
				\node[draw,circle,thick,scale=1,fill=black,label=above:{}] (1) at (0,0){};
				\node[draw,circle,thick,scale=1,fill=black,label=above:{}] (2) at (1,0){};
				\node[draw,circle,thick,scale=1,fill=black,label=above:{}] (3) at (2,0){};
				\node[draw,circle,thick,scale=1,fill=black,label=above:{}] (5) at (0,1){};
				\node[draw,circle,thick,scale=1,fill=black,label=above:{}] (4) at (1,1){};
				\draw[thick] (1) to (2) to (3) to (4) to (5)  to (1) to (6) to (4);
				\draw[thick] (1) to (4) to (2);
			\end{tikzpicture}};
					\node (2) at (5,0)	{\begin{tikzpicture}
				\node[draw,circle,thick,scale=1,fill=black,label=above:{}] (6) at (0,2){};
				\node[draw,circle,thick,scale=1,fill=black,label=above:{}] (1) at (0,0){};
				\node[draw,circle,thick,scale=1,fill=black,label=above:{}] (2) at (1,0){};
				\node[draw,circle,thick,scale=1,fill=black,label=above:{}] (3) at (2,0){};
				\node[draw,circle,thick,scale=1,fill=black,label=above:{}] (5) at (0,1){};
				\node[draw,circle,thick,scale=1,fill=black,label=above:{}] (4) at (1,1){};
				\draw[thick] (1) to (2) to (3) to (4) to (5)  to (1) to (6) to (4);
				\draw[thick] (2) to (4);   \draw[thick] (5) to (2);
			\end{tikzpicture}};
			\node (3) at (10,0) 
			{
			\begin{tikzpicture}
				\node[draw,circle,thick,scale=1,fill=black,label=above:{}] (6) at (0,2){};
				\node[draw,circle,thick,scale=1,fill=black,label=above:{}] (1) at (0,0){};
				\node[draw,circle,thick,scale=1,fill=black,label=above:{}] (2) at (1,0){};
				\node[draw,circle,thick,scale=1,fill=black,label=above:{}] (3) at (2,0){};
				\node[draw,circle,thick,scale=1,fill=black,label=above:{}] (5) at (0,1){};
				\node[draw,circle,thick,scale=1,fill=black,label=above:{}] (4) at (1,1){};
				\draw[thick] (1) to (2) to (3) to (4) to (5)  to (1) to (6) to (4);
				\draw[thick] (5) to (3);
				\draw[thick] (5) to (2);
			\end{tikzpicture}};
			\node (4) at (5,5)	{\begin{tikzpicture}
				\node[draw,circle,thick,scale=1,fill=black,label=above:{}] (6) at (0,2){};
				\node[draw,circle,thick,scale=1,fill=black,label=above:{}] (1) at (0,0){};
				\node[draw,circle,thick,scale=1,fill=black,label=above:{}] (2) at (1,0){};
				\node[draw,circle,thick,scale=1,fill=black,label=above:{}] (3) at (2,0){};
				\node[draw,circle,thick,scale=1,fill=black,label=above:{}] (5) at (0,1){};
				\node[draw,circle,thick,scale=1,fill=black,label=above:{}] (4) at (1,1){};
				\draw[thick] (1) to (2);
				\draw[thick] (2) to (3);
				\draw[thick]  (3) to (4);
				\draw[thick]  (4) to (2);
				
							\draw[thick]  (1) to (6) to (4);
				\draw[thick] (2) to (5);   \draw[thick] (6) to (2);
			\end{tikzpicture}};
			\end{tikzpicture}}
			\end{center}
			\caption{Flops between the four crepant resolutions of the singularity $k[x,y,z,t]/(t^2- x y z )$. 
			This binomial variety is the double cover of $\mathbb{A}^3$ branched along the three coordinate axes.
			This singular variety is a $\mathbb{Z}_2\times\mathbb{Z}_2$ orbifold of $\mathbb{C}^3$. \label{Fig:C322}
			 }
			\end{figure}

			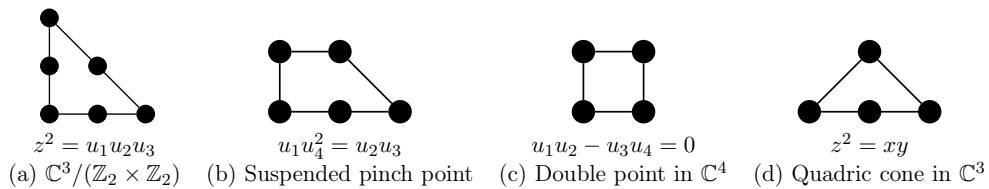
\begin{figure}[hbt]
 \begin{center}
 \scalebox{.8}{
 \begin{tabular}{cccc}
\scalebox{.8}{
 \begin{tikzpicture}
				% F4
				\node[draw,circle,thick,scale=1,fill=black,label=above:{}] (6) at (0,2){};
				\node[draw,circle,thick,scale=1,fill=black,label=above:{}] (1) at (0,0){};
				\node[draw,circle,thick,scale=1,fill=black,label=above:{}] (2) at (1,0){};
				\node[draw,circle,thick,scale=1,fill=black,label=above:{}] (3) at (2,0){};
				\node[draw,circle,thick,scale=1,fill=black,label=above:{}] (5) at (0,1){};
				\node[draw,circle,thick,scale=1,fill=black,label=above:{}] (4) at (1,1){};
							\draw[thick] (1) to (3) to (6) to (1);
			\end{tikzpicture}}
 &
 
 \begin{tikzpicture}
				\node[draw,circle,thick,scale=1,fill=black,label=above:{}] (1) at (0,0){};
				\node[draw,circle,thick,scale=1,fill=black,label=above:{}] (2) at (1,0){};
				\node[draw,circle,thick,scale=1,fill=black,label=above:{}] (3) at (2,0){};
				\node[draw,circle,thick,scale=1,fill=black,label=above:{}] (5) at (0,1){};
				\node[draw,circle,thick,scale=1,fill=black,label=above:{}] (4) at (1,1){};
				\draw[thick] (1) to (3) to (4) to (5) to (1);
			\end{tikzpicture}
 &
 \begin{tikzpicture}
				\node[draw,circle,thick,scale=1,fill=black,label=above:{}] (1) at (0,0){};
				\node[draw,circle,thick,scale=1,fill=black,label=above:{}] (2) at (1,0){};
				\node[draw,circle,thick,scale=1,fill=black,label=above:{}] (5) at (0,1){};
				\node[draw,circle,thick,scale=1,fill=black,label=above:{}] (4) at (1,1){};
				\draw[thick] (1) to (2) to (4) to (5) to  (1);
			\end{tikzpicture}
 & 
 \begin{tikzpicture}
				
				\node[draw,circle,thick,scale=1,fill=black,label=above:{}] (1) at (0,0){};
				\node[draw,circle,thick,scale=1,fill=black,label=above:{}] (2) at (1,0){};
				\node[draw,circle,thick,scale=1,fill=black,label=above:{}] (3) at (2,0){};
				\node[draw,circle,thick,scale=1,fill=black,label=above:{}] (4) at (1,1){};
				\draw[thick] (1) to (3) to (4) to (1);
			\end{tikzpicture}\\
			 $z^2=u_1 u_2 u_3$ & $u_1 u_4^2 =u_2 u_3$ & $u_1 u_2 -u_3 u_4=0$ & $z^2 = xy $\\
			 (a)   $\mathbb{C}^3/(\mathbb{Z}_2\times\mathbb{Z}_2)$& (b) Suspended pinch point & (c) Double point in $\mathbb{C}^4$ & (d)  Quadric cone in $\mathbb{C}^3$ \\
 \end{tabular}
 }
 \end{center}
\caption{Some of the singular binomial varieties encountered in this paper.  
}
\end{figure}

\clearpage

\subsection{SPP flopping between Y$_4$, Y$_5$, Y$_6$, and Y$_8$}

The  flops connecting  Y$_4$, Y$_5$, and Y$_6$ can be easily understood using the tree of blowups given in equation 
\eqref{Ch456}.  Consider the proper transform $Y$ of the Weierstrass equation after the four blowups defining $X_4$ in equation \eqref{Ch456}. 
By definition, $Y$ is a common blowdown of Y$_4$, Y$_5$, and Y$_6$. It can be written in the following suggestive form  (see equation \eqref{Y6Res}):
\begin{equation}
Y:\quad e_2( e_3 y^2-a e_1  s^3 x+b e_1^2 e_2^2 s^5)- e_3  e_4^2  (e_1  x^3)=0.
\end{equation}
In the patch $e_1 x ys\neq 0$, it has the following structure 
\begin{equation}
Y :\quad e_2  q=e_3 e_4^2 , \quad q=(e_3 y^2-a e_1  s^3 x+b e_1^2 e_2^2 s^5)/(e_1 x^3).
\end{equation}
We recognize this as a suspended pinch point singularity at the origin $e_2 =q=e_3=e_4=0$. Such a singularity admits three crepant resolutions described by the tail of the tree of blowups in \eqref{Ch456} and mimic the blowups in Figure \ref{Fig:SSPRes}. 

The same story hold for the flops between  the minimal models Y$_8$, Y$_5$, and Y$_4$ by considering the proper transform of the Weierstrass model after the first five blowups defined in the tree of equation 
\eqref{Ch458}. After these five blowups, the proper transform of the Weierstrass model of an E$_7$-model takes the following form (see equation \eqref{eq:Y8}: 
\begin{equation}
e_4  y^2-e_1 e_2 e_3 (e_2 e_4 e_5^2  x^3+a e_1 s^3 x+b e_1^2 e_3  s^5)=0.
\end{equation}
Since when $e_4=y=0$, $e_1e_2$ is nonzero by the projective coordinates in equation \eqref{eq:Y8projcoord}. Thus, we again have an equation of the form 
\begin{equation}
e_4 y^2 =e_3 A.
\end{equation}
We recognize the form of the binomial variety of a suspended pinch point. We then derive $Y_4$, $Y_5$, and $Y_5$ as its three crepant resolutions following the discussion of section  \ref{sec:SSP}.
The last three blowups in each branch of the tree presented in equation \eqref{Ch458} mimic the resolutions of the suspended pinch points as given in section  \ref{sec:SSP}.

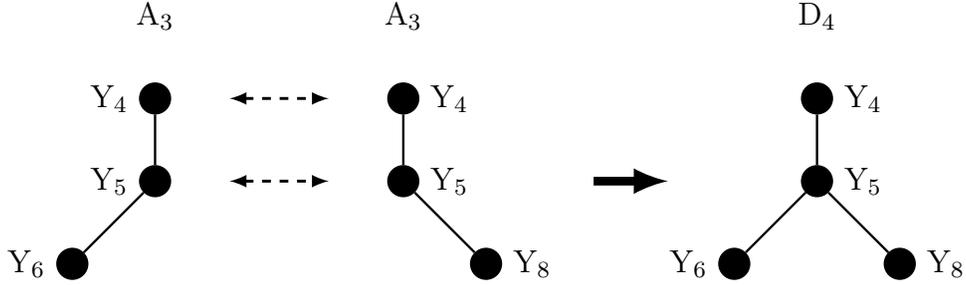
\begin{figure}[h!]
\begin{center}
			\scalebox{1.1}{\begin{tikzpicture}
					\node (1) at (0,0)	{\begin{tikzpicture}
				\node[draw,circle,thick,scale=1,fill=black,label=left:{Y$_6$}] (1) at (-2,0){};
				\node[draw,circle,thick,scale=1,fill=black,label=left:{Y$_5$}] (2) at (-1,1){};
				\node[draw,circle,thick,scale=1,fill=black,label=left:{Y$_4$}] (3) at (-1,2){};
				\draw[thick] (1) to (2) to (3);
				\node[draw,circle,thick,scale=1,fill=black,label=right:{Y$_8$}] (4) at (3,0){};
				\node[draw,circle,thick,scale=1,fill=black,label=right:{Y$_5$}] (5) at (2,1){};
				\node[draw,circle,thick,scale=1,fill=black,label=right:{Y$_4$}] (6) at (2,2){};
				\draw[thick] (4) to (5) to (6);
				\node[draw,circle,thick,scale=1,fill=black,label=left:{Y$_6$}] (7) at (6,0){};
				\node[draw,circle,thick,scale=1,fill=black,label=right:{Y$_5$}] (8) at (7,1){};
				\node[draw,circle,thick,scale=1,fill=black,label=right:{Y$_4$}] (9) at (7,2){};
				\node[draw,circle,thick,scale=1,fill=black,label=right:{Y$_8$}] (10) at (8,0){};
				\draw[thick] (7) to (8) to (9); \draw[thick] (8) to (10);	
				\draw[dashed, latex-latex, line width=1 pt] (-.1,1)--(1.1,1);
								\draw[dashed, latex-latex, line width=1 pt] (-.1,2)--(1.1,2);
								\draw[ -latex, line width=3 pt] (4.3,1)--(5.2,1);
								\node (11) at (-1,3){A$_3$};
								\node (13) at (2,3){A$_3$};
																\node (14) at (7,3){D$_4$};

			\end{tikzpicture}};		
\end{tikzpicture}}
\end{center}
\caption{
The triples ($Y_4, Y_5,Y_8$) and ($Y_4, Y_5,Y_6$) both form an A$_3$ Dynkin diagram where the nodes are the varieties and two nodes are connected if they are related by a flop as shown on the left and middle pictures. For each of the triples, the flops mirror those of a suspended pinch point.  By identifying $Y_4$ and $Y_5$ of the two triples with themselves, we end up gluing the two A$_3$ to form a D$_4$ Dynkin diagram representing all the flops between $Y_4$, $Y_5$, $Y_6$, and $Y_8$.
\label{Fig:A3A4D4}
}
\end{figure}

\section{ Application to $\mathcal{N}=1$ five-dimensional theories}
M-theory compactified on a Calabi-Yau threefold results in a five-dimensional supersymmetric gauge theory with eight supercharges that we denote 
${\cal N}=1$ five-dimensional supergravity. 
Such a theory contains a gravitational multiplet, n$_T$ tensor multiplets, n$_V$ vector multiplets, and n$_H$ hypermultiplets. 
In five-dimensional spacetime, a massless tensor  multiplet is dual to a massless vector. 
In what follows, we assume that all tensors are massless and are dualized to vectors.  
Each vector multiplet contains a real scalar field $\phi$ and each hypermultiplet contains a quaternion (four real fields). 
 The kinetic terms of all the vector multiplets and the graviphoton as well as the Chern-Simons terms are determined by a real function of the scalar fields called the prepotential $\mathscr{F}(\phi)$.

 In the Coulomb branch of an $\mathcal{N}=1$ supergravity theory in five dimensions, 
 the scalar fields of the vector multiplets take values in the Cartan sub-algebra of the Lie algebra and the Lie group is broken to a product of  Abelian factors. This implies that the charge of a hypermultiplet is simply a  
 weight of the representation under which it transforms \cite{IMS}.   In the presence of hypermultiplets charged under a representation $\mathbf{R}$ of the gauge group, the Coulomb phase of the theory is characterized by a  one-loop quantum correction to the  superpotential derived by integrating out massive hypermultiplets.  
  The full quantum superpotential $\mathcal{F}(\varphi)$ is  protected  from further corrections by supersymmetry and is  a piecewise cubic polynomial  of the scalar fields. The metric of the scalar fields of the vector multiplets is the matrix of second derivatives of $\mathcal{F}(\varphi)$ and it is differentiable  in open regions that define distinct Coulomb branches separated by walls on which some of the massive hypermultiplets become massless and should be added to the low energy description of the theory. 
The Intrilligator-Morrison-Seiberg (IMS) prepotential is the quantum contribution to the prepotential of a five-dimensional gauge theory after integrating out all massive fields. 

Let $\phi$ be in the Cartan subalgebra of a Lie algebra $\mathfrak{g}$. We denote by  $\mathbf{R}_i$ the representations under which the hypermultiplets transform. 
The  weights are in the dual space of the Cartan subalgebra.  We denote the evaluation of a  weight on a coroot vector $\phi$ as a scalar product $\langle \mu,\phi \rangle$.  Denoting the roots by $\alpha$ and the weights of $\mathbf{R}_i$ by $\varpi$ we have \cite{IMS}
\begin{align}
6\mathscr{F}_{\text{IMS}}(\phi) =&\frac{1}{2} \left(
\sum_{\alpha} |\langle \alpha, \phi \rangle|^3-\sum_{i} \sum_{\varpi\in \mathbf{R}_i} n_{\mathbf{R}_i} |\langle \varpi, \phi\rangle|^3 
\right).\label{Eq:IMS}
\end{align}
For all simple groups with the exception of SU$(N)$ with $N\geq 3$, this is the full purely cubic sector of the prepotential as there are no non-trivial third Casimir invariants. 

The {\em open dual fundamental Weyl chamber} is defined as the cone $ \langle \alpha, \phi \rangle>0$, where $\alpha$ runs through the set of all simple positive roots. 
This choice makes it possible to remove the absolute values in the sum over the roots. 
 For a given choice of a group $G$ and representations $\mathbf{R}_i$, 
 we then consider the  hyperplane arrangement  $\langle \varpi, \phi\rangle=0$, where $\varpi$ 
  runs through all the weights of all the representations $\mathbf{R}_i$ and $\phi$ is an element of the coroot space. 
 If none of these hyperplanes intersect the interior of the dual Weyl chamber of $\mathfrak{g}$, we can safely remove the absolute values in the sum over the weights. 
 Otherwise, we have hyperplanes partitioning the dual fundamental Weyl chamber into subchambers. Each of these subchambers is defined by the signs of the linear forms $\langle \varpi, \phi\rangle$ and corresponds to a specific sector of the Coulomb branch. 
 Two such subchambers are adjacent when they differ by the sign of a unique linear form.  
 Within each of these subchambers, the prepotential is a cubic polynomial. 
 But as we go from one subchamber to an adjacent one, we have to go through one of the walls defined by the weights. 
 The transition from one chamber to an adjacent chamber is a phase transition.

\subsection{Prepotential and Coulomb phases}   
 It is immediate to compute the prepotential for a gauge theory with gauge group E$_7$ coupled to  $n_A$ hypermultiplets tranforming in the adjoint representation and $n_F$ hypermultiplets transforming in the fundamental representation $\bf{56}$.

 First we recall that the open dual fundamental  Weyl chamber is in our conventions given by the seven inequalities  $\langle \alpha_i, \phi\rangle$ for $i=1,\ldots, 7$:
 \begin{equation}
 \begin{cases}
 2 \phi _1-\phi _2>0, \quad 
 -\phi _1+2 \phi _2-\phi _3>0,\quad
 -\phi _2+2 \phi _3-\phi _4-\phi _7>0,\quad
 -\phi _3+2 \phi _4-\phi _5>0,\\
 -\phi _4+2 \phi _5-\phi _6>0,\quad 
 2 \phi _6-\phi _5>0,\quad
 2 \phi _7-\phi _3>0.
 \end{cases}
 \end{equation}
Each of the eight chambers of I(E$_7$, $\bf{56})$ is uniquely defined by the signs taken by the seven linear functions $\langle \varpi_i, \phi\rangle$ for $i=\{19, 20,23,26,29,32,30\}$. 
These are classified in Table \ref{Table:Ch}. Imposing these signs together with the condition of being inside the open dual Weyl chamber defines each chamber. 
We can then compute the prepotential for each chamber:  as by their very definition  they resolve the absolute values in the definition of $\mathscr{F}(\phi)$, we end up with polynomials.   We illustrate the process for Chamber 5. 
The signs defining Chamber 5 are  $(+,+,+, +,-,-,-)$. 
\begin{equation}
\begin{cases}
\phi _1-\phi _6>0\quad 
-\phi _1+\phi _2-\phi _6>0\quad
-\phi _2+\phi _3-\phi _6>0\quad
-\phi _3+\phi _4-\phi _6+\phi _7>0\\
-\phi _4+\phi _5-\phi _6+\phi _7<0\quad
\phi _7-\phi _5<0\quad
\phi _4-\phi _6-\phi _7<0
\end{cases}
\end{equation} 
We can then immediately compute the  prepotential $\mathscr{F}_5(\phi)$:
 \begin{equation}
 \begin{aligned}
6\mathscr{F}_{5}(\phi) &=
8 (1-n_A) (\phi _1^3+ \phi _2^3+\phi _3^3+\phi _5^3)+(8 (1-n_A)-2n_F)( \phi _4^3+\phi _6^3+ \phi _7^3 )\\
& +6 (1-n_A-n_F)(\phi _2 \phi _1^2- \phi _2 \phi _3^2)+6(2 n_A+n_F-2)  \phi _2^2 \phi _1 \\
&+ 6(n_A-n_F-1) \left(\phi _3^2\phi _4
+ \phi _3^2\phi _7-2 \phi_5 \phi_6^2\right)+6 \phi _5^2 \phi _6 (3 n_A-2 n_F-3)\\
&+6 (1-n_A) \phi _4 \left(\phi_5^2 -2\phi_4 \phi _5\right)-6n_F\phi _4^2 \left(  \phi _6+  \phi _7\right) +12 n_F \phi_4\phi _6 \phi _7\\
&-6 n_F \left(\phi _3 \phi _2^2-\phi _3 \phi _4^2-\phi _3 \phi _7^2+\phi _6 \phi _7^2+\phi _6^2 \phi _7-2 \phi _3 \phi _4 \phi _7   -2\phi_4 \phi _6 \phi _5+\phi_4 \phi _6^2+\phi_4 \phi _7^2
\right).
\end{aligned}
 \end{equation}

\subsection{Counting charged hypermultiplets in $5d$ using triple intersection numbers}
We can compute the number of charged hypermultiplets by comparing the triple intersection numbers with $6 \mathscr{F}_{IMS}(\varphi)$. 

We illustrate the process in Chamber 5. We can compute the number of multiplets in the adjoint and the fundamental representations by just computing the intersection numbers $D_1^3$ and $D_4^3$ as the coefficients of $\phi_1^3$ and $\phi_4^3$ depends on linearly on $n_A$ and $n_F$ and are supposed to match the intersection numbers $D_1^3$ and $D_4^3$. Since $D_1$ is a ruled surface over a curve of genus $g$ and $D_4$ is the blowup at $S\cdot V(a)$ of a ruled surface over a curve of genus $g$  we have $D_1^3= 8(1-g)$ and $D_4^3= 8(1-g)-S\cdot V(a)$. 
  Thus 
\begin{equation}
n_{A}=g, \quad n_F= \frac{1}{2}S\cdot V(a).
\end{equation}
Since the class of $a$ is  $[a]=-4K-3S$, and $KS=2g-2 -S^2$, we conclude that 
\begin{equation}\label{Eq.NAF}
n_{A}=g, \quad n_F= \frac{1}{2}\Big(8(1-g)+S^2\Big).
\end{equation}
We can apply the same technique with the same result in the other chambers as the number of multiplets does not change between different phases of the  Coulomb branch. 

In the next section, we will see that the same numbers $n_A$ and $n_F$ are required to cancel anomalies of a ${\cal N}=(1,0)$ six-dimensional theory obtained from a compactification of  F-theory on a  Calabi-Yau variety that is a crepant resolution of the Weierstrass model of an E$_7$-model. 
\section{Application to $\mathcal{N}=(1,0)$ six-dimensional theories}
Local anomalies in six-dimensional theories can be used to constrain the number of charged hypermultiplets. In the best cases, anomaly conditions can even completely fix the matter content of a given six-dimensional theory. 
This was already  brilliantly illustrated by Seiberg just after the first string revolution in a paper in which the absence of anomalies was used to derive the number of matter multiplets in a six-dimensional superconformal field theory~\cite{Seiberg:1988pf}. 

\subsection{Anomaly cancellations in $\mathcal{N}=(1,0)$ six-dimensional theories}

Here, we consider a $(1,0)$ supersymmetric gauged supergravity theory. Such theories have chiral  spinors and chiral tensors and have potential  
gravitational, gauged, and mixed local anomalies. These local anomalies can be cancelled by the Green-Schwarz mechanism. 
 F-theory compactified on a Calabi-Yau threefold gives a $(1,0)$ supersymmetric six-dimensional gauged supergravity for which the Green-Schwarz mechanism was studied by Sadov~\cite{Sadov:1996zm}, see also~\cite{GM1,Park,  Monnier:2017oqd}. 
 In what follows, $R$ represents the Riemann curvature of spacetime and $\tr R^n$ are computed with respect to the fundamental representation of SO(5,1). We denote by $n_T$, n$_V^{(6)}$, and n$_H$ the number of tensor, vector, and hypermultiplets, respectively. The gauge group is $G$ and charged hypermultiplets transform in representations $\bf{R}_i$ of the gauge group. 
 We can distinguish two types of hypermultiplets: those that are neutral and those that transform under some representation $\bf{R}_i$ of the gauge group. 
 In our conventions, hypermultiplets that are charged under a zero weight of the representation are considered neutral.

The relevant trace identities for E$_7$ are \cite{Erler}
\begin{equation}
\tr_{\bf{133}} F^2=3\tr_{\bf{56}} F^2, \quad \tr_{\bf{133}} F^4=\frac{1}{6}(\tr_{\bf{56}} F^2)^2, \quad 
\tr_{\bf{56}} F^4=\frac{1}{24}(\tr_{\bf{56}} F^2)^2,
\end{equation}
where we note that
 E$_7$ does not have an independent fourth Casimir invariant. 
The trace identities for a representation $\mathbf{R}_{i}$ of a simple group $G$ are
\begin{equation}
\tr_{\bf{R}_{i}} F^2_a=A_{\bf{R}_{i}} \tr_{\bf{F}} F^2 , \quad \tr_{\bf{R}_{i}} F^4=B_{\bf{R}_{i}} \tr_{\bf{F}} F^4+C_{\bf{R}_{i}} (\tr_{\bf{F}} F^2 )^2,
\end{equation}
with respect to a  reference representation $\bf{F}$ that we can freely choose. We take it to be the fundamental representation. 
These  group theoretical coefficients are listed in \cite{Erler}.

The anomalies are canceled by the Green-Schwarz mechanism when I$_8$ factorizes \cite{Green:1984bx,Sagnotti:1992qw,Schwarz:1995zw}. 
The modification of the field strength $H$ of  the antisymmetric tensor $B$ is 
\begin{equation}
H= dB + \frac{1}{2} K \omega_{3L}+ \frac{2}{\lambda}S\omega_{a,3Y}, 
\end{equation}
where  $\omega_{3L}$ and $\omega_{a,3Y}$ are respectively the gravitational and Yang-Mills Chern-Simons  terms. 
 If I$_8$ factors as 
 \begin{equation}
 \text{I}_8= X\cdot  X,
 \end{equation}
  then the anomaly is canceled by adding the following Green-Schwarz counter-term 
\begin{equation}
\Delta L_{GS}\propto \frac{1}{2} B \wedge X.
\end{equation}
Following Sadov \cite{Sadov:1996zm}, to cancel the local anomalies, we consider
\begin{equation}
X= \frac{1}{2} K \tr R^2 + \frac{2}{\lambda} S\tr F^2,
\label{eq:Xfactor}
\end{equation} where  $\lambda$ is a constant normalization factor chosen such that the  smallest
topological charge of an embedded SU($2$) instanton in the gauge group G is one \cite{Kumar:2010ru, Park, Bernard}. 
This forces $\lambda$ to be the Dynkin index of the fundamental representation of  $G$ as listed in Table \ref{tb:normalization} \cite{Park}. 

If the simple group $G$ is supported on a divisor $S$, the local anomaly cancellation conditions are  the following equations \cite{Sadov:1996zm,Kumar:2010ru, Park, Bernard}
\begin{subequations}
\begin{align}
n_T&=9-K^2 , \\
n_H-n_V^{(6)}+29n_T-273 &=0,\\
\left(B_{\bf{adj}}-\sum_{i}n_{\bf{R}_{i}}B_{\bf{R}_{i}}\right)& = 0, \\
\lambda  \left(A_{\bf{adj}}-\sum_{i}n_{\bf{R}_{i}}A_{\bf{R}_{i}}\right) & =6  K\cdot S, \\
\lambda^2 \left(C_{\bf{adj}}-\sum_{i}n_{\bf{R}_{i}}C_{\bf{R}_{i}}\right) & =-3 S ^2.
\end{align}
\end{subequations}
The first equation gives us the number of tensor multiplets $n_T$. 
The second equation is the vanishing of the coefficient of $\tr R^4$ and will be checked at the end.
The third equation is automatically satisfied since  E$_7$  has no independent quartic Casimir (and therefore all the coefficients $B_{adj}$ and $B_i$ are zero). 

\begin{table}[htb]
\begin{center}
\begin{tabular}{|c|c|c|c|c|c|c|c|c|c|}
\hline
 $\mathfrak{g}$ & A$_n$ & B$_n$ & C$_n$ & D$_n$ & E$_8$ & E$_7$ & E$_6$&  F$_4$ & G$_2$ \\
 \hline
 $\lambda$ & $1$ & $2$  & $1$ & $2$ & $60$ & $12$ & $6$ & $6$ & $2$ \\
 \hline  
\end{tabular}
\caption{The normalization factors for each simple gauge algebra. See \cite{Kumar:2010ru}.}
\label{tb:normalization}
\end{center}
\end{table}

The total number of hypermultiplets is the sum of the neutral hypermultiplets and the charged hypermultiplets. 
For a compactification on a Calabi-Yau threefold $Y$, the  number of  neutral hypermultiplets is  $n_H^0=h^{2,1}(Y)+1$ \cite{Cadavid:1995bk}. The number of each multiplet is \cite{Cadavid:1995bk,GM1}
\begin{align}
n_V^{(6)}&=\dim{G}, \quad n_T=h^{1,1}(B)-1=9-K^2 , \\
n_H&=n_H^0+n_H^{ch}=h^{2,1}(Y)+1+\sum_{i} n_{\bf{R}_i} \left( \dim{\bf{R}_i}-\dim{\bf{R_{i}^{(0)}}} \right),
\end{align}
where the (elliptically-fibered) base $B$ is a rational surface with canonical class $K$. Here,   n$_{\bf{R}_i}$ is the number of hypermultiplets transforming non-trivially in the representation $\bf{R}_i$ of the gauge group, $\dim{\bf{R_{i}^{(0)}}}$ is the number of zero weights in the representation $\bf{R}_i$.  

\subsection{Counting hypermultiplets in $6d$ using anomaly cancellation conditions} 

With the gauge group E$_7$, we have $n_V^{(6)}=133$ and $h^{2,1}(Y)$ was computed in \cite{Euler} for any crepant resolution of a Weierstrass model of an E$_7$-model: 
\begin{equation}\label{eq:Hodge}
n_V^{(6)}=133, \quad h^{2,1}(Y)=18 + 29K^2 + 49K\cdot S + 21S^2.
\end{equation}
For E$_7$, the normalization factor $\lambda$ is $12$. 
Since we have a unique component in the gauge group, the anomaly equations are:
\begin{equation}
12 (3-3n_{\bf{133}}-n_{\bf{56}})= 6 K \cdot S, \quad 
12^2 (\frac{1}{6}-\frac{1}{6} n_{\bf{133}}-\frac{1}{24} n_{\bf{56}})=-3 S^2.
\end{equation}
Using the identity $2-2g=-K\cdot S -S^2$, where $g$ is the arithmetic genus of the curve $S$, we find:
\begin{equation}\label{eq:NH6}
n_{\bf{133}}=g, \quad n_{\bf{56}}=\frac{1}{2}(8-8g+ S^2),
\end{equation}
{which matches what we found by comparing the triple intersection numbers with the prepotential in the ${\cal N}=1$ five-dimensional theory obtained by compactifiying M-theory on the same elliptic fibration.}

 The total number of hypermultiplets is then 
  \begin{equation}\label{eq:nH}
n_H=(18 + 29K^2 + 49K\cdot S + 21S^2)+1+ (133-7) n_{\bf{133}}+56 n_{\bf{56}}= 29(5+K^2).
 \end{equation}
Using equations \eqref{eq:NH6} and \eqref{eq:nH}, it is a simple arithmetic computation to check that the coefficient of $\tr R^4$ vanishes identically \cite{Salam}:
\begin{equation}
n_H-n_V^{(6)}+29n_T-273=0.
\end{equation}

\begin{rem}
As explained in the introduction, the value of $n_{\bf{56}}$ computed in \eqref{eq:NH6} equals one-half of the intersection number $S\cdot V(a)$, 
which is the number of points over which the fiber III$^*$ degenerates to an incomplete fiber of type II$^*$. 
The number is a half-integer when $S$ has odd self-intersection. 
If $S$ is a rational curve, we see that requesting $n_{\bf{56}}$ to be non-negative forces the self-intersection number $S^2$ to be greater or equal to $-8$. 
A rational curve with self-intersection $-n$ $(n>0)$ is called a $-n$-curve. 
When $S$ is a $-8$-curve, $S$ does not intersect $V(a)$ and there is no matter in the representation $\bf{56}$. 
When $S$ is a $-7$-curve, $S$ intersects $V(a)$ at exactly one point and we get one half-hypermultiplet in the representation $\bf{56}$. 
\end{rem}

\section*{Acknowledgements}
 M.E. is supported in part by the National Science Foundation (NSF) grant DMS-1701635 ``Elliptic Fibrations and String Theory.''  S.P. is supported by the National Science Foundation through a Graduate Research Fellowship under grant DGE-1144152 and by the Hertz Foundation through a Harold and Ruth Newman Fellowship.

\appendix
\section{Relevant Definitions}
\begin{defn}
An {\em elliptic fibration} is a projective surjective morphism $\varphi: X\longrightarrow B$ between normal varieties such that $\varphi$  is proper  and endowed with  a rational section\footnote{Given a surjective morphism $\varphi:X\to B$, a rational section is a rational map $\sigma:B\longrightarrow X$ such that 
$\varphi\circ \sigma$ is the identity away from a Zariski closed set of $X$.}
and the  fiber  over the generic point of $B$ is a regular projective curve of genus one. 
\end{defn}

\begin{defn}
The {\em discriminant locus} of an elliptic fibration  $\varphi: X\longrightarrow B$ is the locus of points $p$ such that the fiber over $p$ is a singular fiber. 
Under mild assumptions, the discriminant locus is a divisor in the base. 
The prime components of the pre-image of prime components of the discriminant locus are called {\em fibral divisors}. 
\end{defn}

\begin{defn}[$\mathbb{Q}$-Gorenstein, $\mathbb{Q}$-factorial, nef-divisors]
\quad 
\begin{enumerate}
\item A prime Weil divisor $D$ is said to be  $\mathbb{Q}$-Cartier   when some integral multiple of $D$ is a Cartier divisor. 
\item A variety $Y$ is said to be {\em $\mathbb{Q}$-Gorenstein} if its canonical class is $\mathbb{Q}$-Cartier. 
\item A normal variety $Y$ is said to be {\em $\mathbb{Q}$-factorial} if every prime Weil divisor $D$ on $X$ is  $\mathbb{Q}$-Cartier. \item Let  $f:{Y}\longrightarrow X$ be a  projective morphism  between normal varieties. A divisor $D$ is said to be {\em $f$-nef  (or relatively nef over $S$)} when $\int_Y D\cdot C \geq 0$ for any curve $C$ projected to a point by $f$.
\end{enumerate}
\end{defn}

\begin{defn}[Minimal models]
A {\em minimal model} of  a projective morphism $f:Y\longrightarrow S$ between quasi-projective normal varieties is a projective morphism $f': Y'\longrightarrow S$ such that 
\begin{enumerate}
\item There exists a birational map $\alpha: Y'\dasharrow Y$ such that $f'=f\circ \alpha$.
\item $Y'$ has only $\mathbb{Q}$-factorial terminal singularities. 
\item $K_{Y'}$ is $f'$-nef.
\end{enumerate}
\end{defn}
\begin{defn}[Modifications, (crepant) resolutions of singularities]\quad
\begin{enumerate}
\item A {\em modification} of a variety $X$ is a proper birational  morphism $f:\widetilde{X}\longrightarrow X$. 
\item A  {\em resolution} of a singularity is a modification $f:\widetilde{X}\longrightarrow X$ such that $\widetilde{X}$ is nonsingular. 
\item A resolution $f:\widetilde{X}\longrightarrow X$  is said to be {\em crepant} when $X$ is  $\mathbb{Q}$-Gorenstein and $f$ preserves the canonical class of $X$ ($K_{\widetilde{X}}=\varphi^* K_X$). 
\item A {\em strong resolution} is a resolution that is an isomorphism away from the singular locus, that is, a resolution $f:\widetilde{X}\longrightarrow X$ such that the restriction of $f$ to $\widetilde{X}\backslash f^{-1}(Sing(X))$  is an isomorphism onto  $X\backslash Sing(X)$. 

\end{enumerate}
\end{defn}
In this paper we only consider strong resolutions and simply call them {\em resolutions}. 
\begin{rem}
 A crepant resolution of $Y$ is automatically a nonsingular minimal model over $Y$.  
Distinct crepant resolutions are connected by a chain of flops if we assume the existence and the termination of flops. 
Kawamata has shown~\cite{Kawamata} that the seminal results of 
Birkar-Cascini-Hacon-McKernan~\cite{BCHM} together with the boundedness
of length of extremal rays imply that different minimal models are  connected by a sequence of flops. 

\end{rem}

 A variety can be $\mathbb{Q}$-factorial in the projective  category but fail to be  $\mathbb{Q}$-factorial in the analytic  category. 
An important consequence of the following lemma is that terminal $\mathbb{Q}$-factorial singularities are obstructions to the existence of a crepant resolution. 
\begin{lem}[{See \cite[Corollary 2.63]{KM}}] 
Let $f:Y\to X$ be a birational morphism between normal varieties. If $X$ is $\mathbb{Q}$-factorial, then the exceptional locus of $f$ is of pure codimension one. 
\end{lem}

\end{document}